\documentclass[12pt]{article}

\usepackage{amssymb,amsmath,amsfonts,amssymb}
\usepackage{graphics,graphicx,color,hhline}
\usepackage{bbm}
\usepackage{stmaryrd}
 \usepackage{mathtools}

\usepackage{authblk}
\usepackage{euscript}

\def\@abssec#1{\vspace{.05in}\footnotesize \parindent .2in 
{\bf #1. }\ignorespaces} 

\graphicspath{{/EPSF/}{Figures/}}   

\newtheorem{theorem}{Theorem}[section]
\newtheorem{lemma}[theorem]{Lemma}

\def\ds{\displaystyle}
\def \Rm {\mathbb R}

\def \NN {\mathbb N}
\def \Zm {\mathbb Z}

\newcommand{\eps}{\varepsilon}

\newcommand{\be}{\begin{equation}}
\newcommand{\ee}{\end{equation}}
\newcommand{\bea}{\begin{eqnarray}}
\newcommand{\eea}{\end{eqnarray}}
\newcommand{\bee}{\begin{eqnarray*}}
\newcommand{\eee}{\end{eqnarray*}}

\def\fref#1{{\rm (\ref{#1})}}

\newcommand{\calE}{\mathcal E}

\newcommand{\cout}[1]{}

\author[1,2]{Randy Bartels \footnote{rbartels@morgridge.org}}
\author[3,*]{Olivier Pinaud \footnote{olivier.pinaud@.colostate.edu}}

 \affil[1]{Morgridge Institute for Research, Madison, WI 53715 USA}
 \affil[2]{Biomedical Engineering Department, University of Wisconsin--Madison, Madison, WI 53715 USA}
 \affil[3]{Department of Mathematics, Colorado State University\\ Fort Collins CO, 80523}

\begin{document}
\title{Analysis and extensions of the Multi-Layer Born method}

\maketitle
\begin{abstract}
Simulating scalar wave propagation in strongly heterogeneous media comes at a steep computational cost, and the widely used approach to simplification—split-step operators—sacrifices accuracy.  The recently proposed multi-layer Born method has sought to resolve that problem, but because it discards evanescent modes, also produces large errors. In this work our main goal is to propose solutions to this critical issue by including evanescent modes in the simulation.  We work in a setting where backscattering can be neglected, allowing us to only calculate forward propagation, and derive two possible schemes. A rigorous mathematical analysis of the numerical errors shows one method is more accurate.  This analysis is also helpful for choosing optimally the discretization parameters. In addition, we propose high order versions of the multi-layer Born method that offer a lower computational cost for a given tolerance.
\end{abstract}

\section{Introduction}

One of the principal methods for obtaining information about a sample by optical means is through microscopy. In a microscope, light is produced by a target object through light emission or scattering that is spatially localized in the object. The light generated then propagates through an optical system, where the optical field is shaped and mapped onto a detector. This propagation is modeled with a Green’s function that is called a point spread function (PSF) for optical imaging, mapping the source point in the object to a point in the image. \cite{mertz2019introduction} Ideally, the PSF is localized tightly in the region around the image point and the PSF is shift-invariant across the image field of view. Such properties allow for high-quality imaging with high spatial resolution if the width of the PSF is small and the strength of the signal is well above the noise in the data. 

In an important class of problems, we seek to image an object through or embedded within a medium that introduces optical distortions to light propagating from the object to the detector. \cite{leith1992imaging, ntziachristos2010going, yoon2020deep} These distortions arise from light propagating through a medium with spatial variations in the relative dielectric permittivity, $\varepsilon(\mathbf{r}) = n^2(\mathbf{r})$, where $n$ is the refractive index (RI) of the medium. \cite{goodman2015statistical, ishimaru2017electromagnetic, carminati2021principles} These spatial variations in the RI, $\delta n(\mathbf{r})  = n(\mathbf{r}) - n_0 $, are relative to the mean $n_0 = \langle n(\mathbf{r}) \rangle$. Here $\langle \cdot \rangle$ represents a spatial average. The key properties of the RI variation are characterized by the autocorrelation $B_n (\mathbf{r}) = \int \delta n (\mathbf{r}’) \, \delta n^* (\mathbf{r}’+\mathbf{r}) \, d^3 \, \mathbf{r}’$. The characteristic length for RI variations is the correlation length of $\ell_c^3 = B_n^{-1}(\mathbf{0}) \,  \int B_n(\mathbf{r}) \, d^3 \, \mathbf{r} $ and the strength is determined by the standard deviation of the RI fluctuations, $\Delta n = \sqrt{B_n(\mathbf{0})}$. Full calculation of the field scattering is computationally expensive, which limits the size of scattering volume that can be simulated with reasonable computational resources. In this paper, we seek to develop a robust and computationally efficient simulation of the scattering of electromagnetic radiation by a spatial variation in RI.

Note that given this aim, we exclude from our discussion turbid media, which present a very difficult case for simulation, with extremely rapid spatial variations ($\ell_c \lesssim \lambda$) and strong variations in the permittivity (large RI changes from the mean, $\Delta n \gtrsim 0.1 $) strongly distorting optical waves and leading to strong reflection of incident light---such as in the case of clouds, white paint, and snow. \cite{carminati2021principles} Here, one must resort to computationally intensive strategies such as finite difference time domain methods, and the size of simulations are severely restricted. \cite{ccapouglu2013computation, elmaklizi2014simulating} 

Our focus in this work is on describing weak scattering regimes, where $\Delta n$ is small and light is largely scattered in the forward direction. \cite{fante1975electromagnetic} The intensity of the back-reflected field from a single RI boundary is proportional to $\Delta n^2/ 4 \, n^2$, which is small when $\Delta n \ll n$. Forward scattered light also largely maintains the incident polarization state, allowing for a scalar wave treatment to provide a suitable description of the wave propagation. Large, isolated particles are dominated by forward scattering. \cite{ishimaru2017electromagnetic} Similarly, when the correlation length of a random medium is long ($\ell_c \gg \lambda$), the regime is referred to as turbulent and backscattering and depolarization are known to be weak. \cite{goodman2015statistical} This regime is relevant to astronomy and cleared tissues, where RI variations vary slowing on the scale of the wavelength and are small in magnitude, producing only minor distortions to the image. 

Optical scattering from small particles redirects the light direction nearly uniformly, making backscattered light significant for samples containing dielectric particles that are small compared to the optical wavelength. However, for a random variation of RI, even for a short correlation length, the propagating field can be largely forward directed. The extent to which the field is scattered in the forward direction can be evaluated from the first moment of the distributions of scattering angles, $g = \langle \cos \theta \, \sigma_d(\theta) \, d \Omega \rangle/\mu_s$, where $\sigma_d(\theta)$ is the differential scattering cross section per unit volume in the random RI medium. The total scattering coefficient, $\mu_s = \int_\Omega \sigma_d(\theta) \, d \Omega$, gives the rate of scattering events with wave propagation. The polar angle, $\theta$, of the scattered field direction relative to the incident field is defined by $\cos \theta = \hat{\mathbf{s}}_o  \cdot \hat{\mathbf{s}}_i $, where  $\hat{\mathbf{s}}_i$ is the unit vector indicating the direction of the incident wave and $\hat{\mathbf{s}}_o$ is the unit vector along the direction of a scattered wave.

 
Experimental characterization of optical scattering reveals that the spatial variation in the RI in biological tissues can be described as a statistically stationary random field. \cite{jacques2008tutorial} The statistics of the RI exhibit fractal behavior, representing the wide range of the variation in the spatial scales of tissue fluctuations, from small strongly scattering organelles within the cytoplasm of cells to large structures within the cell body and the extracellular matrix. \cite{schmitt1996turbulent, schmitt1998optical, xu2005fractal, sheppard2007fractal, wu2007unified, rogers2013modeling} The size of the spatial variation of RI fluctuation ranges from $< 100$ nm to $>$ several $\mu$m. \cite{jacques2011fractal} The strength of the statistical fluctuations are characterized by the RI variance, which is typically on the order of $\Delta n^2 \sim (0.01)^2$, relative to a background RI of $n_0 \sim 1.367$. \cite{xu2011scattering} This weak RI fluctuation admits reliable modeling of light scattering in tissues with first-order scattering over volume in which single scattering events dominate.

The first Born approximation is widely used to model light scattering in the limit of single scattering events. For a scalar field in the Born approximation, the differential scattering cross section per unit volume for a spatially statistically stationary RI fluctuation is given by $ \sigma_d = 2 \, \pi \, k^4 \, \Phi_n(k_s)$. Here $k = 2 \, \pi /\lambda$ is the wavenumber and the scattering vector $\mathbf{k}_s = n_0 \, k \, (\hat{\mathbf{s}}_o - \hat{\mathbf{s}}_i)$ denotes the change in direction of the incident light that has a magnitude of $k_s = 2 \, n_0 \,  k \, \sin(\theta/2)$. \cite{ishimaru2017electromagnetic} The angular distribution of the light scattered by the random medium is related to the power spectrum of the spatial variation in the RI, $\Phi_n(k_s) = (2 \, \pi)^{-3} \, \int B_n(\mathbf{r}) \, \exp(- i \, 2 \, \pi \, \mathbf{k}_s \cdot \mathbf{r} )  \, d V $. This calculation is valid provided that the scattering length $\ell_s \, =\mu_s^{-1}$ is larger than the thickness, $L$, of the volume modeled. In addition, for the first Born approximation to be valid, it has been shown that we must meet the condition $\Delta n^2 \, (n_0 \, k \, L_0)^2 \gg 1$. \cite{ccapouglu2009accuracy, rogers2013modeling} 

Measurements of the random RI distribution of biological tissues are found to exhibit a fractal spatial structure that is well described by the Whittle-Matern (WM) correlation function.  \cite{rogers2013modeling}  A typical correlation length of $\ell_c \sim 500$ nm is used for tissues modeling. Wavelengths ranging from the visible to the near infrared lead to a correlation parameter $X = k \ell_c = 2 \, \pi \, (\ell_c/\lambda)$ in the range of $X \sim \pi – 2 \, \pi$. For representative tissue RI statistics, the product of the scattering coefficient and the correlation length can be written as $\mu_s \, \ell_c \approx 4 \, \Delta n^2 \, X^2 \, \left[ 1 - \left( 1 + 4 \, X^2\right)^{-2} \right]$.  \cite{xu2011scattering} For wave propagation to be well described by a forward scattering scalar model, we require weak backscattering and high anisotropy ($g$ close to one), which implies that $\Delta g = 1-g$ must be a small parameter. Using typical parameters for the WM correlation function, this forward-scattering parameter can be readily computed. As smaller region of RI fluctuations lead to a broader distribution of scattering direction, $\Delta g$ increases with decreasing $X$, but does not depend on the magnitude of the RI perturbation $\Delta n$. In the limiting case, as $X \rightarrow 0$, a typical value of $\Delta g \rightarrow 0.6$ is not strongly forward directed. For typical values of $X$, we find a range of $\Delta g \sim 0.05 – 0.1$. While single scattering events are largely forward directed when $\Delta g$ is small, after $\sim 1/\Delta g$ scattering events, the scattered light has lost all of the initial propagation direction of the input beam; this length scale is called this transport length, $\ell_t = \ell_s/\Delta g$, and beyond this length even highly anisotropic scattering media can no longer be described by forward-scattered light. \cite{carminati2021principles} 


Simulating wave propagation in strongly heterogeneous media over large distances compared to the wavelength of the electromagnetic field is an extremely challenging task. In the scalar case, this task amounts to numerically solving the three-dimensional Helmholtz equation, for which a very fine grid is necessary in order to achieve a reasonable accuracy. This results in large linear systems whose resolution comes at a considerable computational cost. In the context of imaging and inverse problems, where multiple systems have to be solved to, e.g., reconstruct the refractive index or image scatterers, a full simulation of wave propagation is often prohibitive and one has to resort to approximations. 

A classical simplification stems from the directional nature of the propagating wavefield, as in the case of optical imaging, allowing for neglecting backscattering and taking into account only forward propagation. This is how the so-called paraxial wave equation is derived, and is simply a time-dependent two-dimensional Schr\"odinger equation where the time variable is replaced by the variable along the direction of propagation. The full 3D problem is then reduced to a 2D evolution problem that is now numerically tractable. This wave scattering model is often numerically simulated using a split-step operator approach \cite{flatte1975calculation}, that is sequentially applied to slices of the 3D RI variation. In each slice of thickness $\Delta z$, propagation and scattering are decoupled by sequentially applying a phase distortion followed by propagation of the distorted wave through each slice. The phase accumulation due to the RI variation are modeled as a phase screen  $\phi(x,y) \approx k_0 \int_{\Delta z} n(x,y,z) \, d z$ that is localized to a plane in the slice, followed by propagation of the perturbed wave through the slice with a homogeneous medium equal to the uniform background refractive index. These steps are iteratively applied for each thin slice of the medium. 

This split operator approach has been used to study wave propagation through the ocean \cite{flatte1975calculation} and turbulent atmosphere \cite{fleck1976time, knepp1983multiple, spivack1989split, coles1995simulation, roggemann1995method, frehlich2000simulation}; to study turbulent propagation in biological tissues, \cite{schott2015characterization, glaser2016fractal} and birefringent random media \cite{mu2023multislice}; and to model integrated optical devices \cite{van1981beam} and x-ray scattering from nanostructures \cite{li2017multislice}. These models have been used for numerous studies beyond the initial applications for investigating propagation through atmospheric turbulence and astronomical adaptive optics. In the space of tissue optics, they have been used to study beam propagation in tissue-mimicking random RI fields \cite{glaser2016fractal, cheng2019development} and to study the memory effect and correlations in anisotropy scattering media. \cite{schott2015characterization, judkewitz2015translation, arjmand2021three} In addition, the fast computational speed of split-step propagator models has made them useful for solving inverse scattering imaging problems. \cite{maiden2012ptychographic, tian2014multiplexed, kamilov2016optical, chowdhury2019high, wang2021large}

The main issue with such a model is the loss of accuracy away from the axis of propagation because of the approximations intrinsic to the model. This limitation arises from the use of a thin phase plate, which implies small angle waves incident on the slice and the calculation intrinsically neglects any scattering within $\Delta z$. In addition, simplification of the angular spectral propagator with the Fresnel approximation further limits the simulation to the paraxial approximation and neglects contributions of evanescent waves. Unfortunately, these restrictions fail to capture important physics of light propagation in tissues. In particular, the recognition that the RI fluctuation in tissue displays fractal behavior means that there are strong contributions to optical scattering from small structures such as vesicles, mitochondria, peroxisomes, and lysosomes, all of which have a length scale much smaller than the wavelength. These small structures will couple evanescent fields to far-field radiation when in close proximity, which is a frequent occurrence in tissue. As a result, current approaches to large-scale modeling of tissues with split-step operator approaches likely fail to capture some of the scattering properties of tissues.

Several approaches have been developed to resolve these shortcomings. One of those is the multi-layer Born (MLB) method introduced in \cite{chen2020multi}. The main idea is to use an integral formulation of the Helmholtz equation in which backscattering is neglected, and to slice the heterogeneous medium along the axis of propagation into layers with a width small compared to either the wavelength or the typical oscillation length of the background, whichever is shorter. This condition on the layers' width is necessary to achieve a decent accuracy. This results in a very simple numerical scheme whose computational cost is inversely proportional to the layers' width, and which consists in evaluating iteratively 2D convolutions of the form, for $z>0$, 
\be \label{deff}
f(x)=\int_{\Rm^2} G(x-y,z) U(y) dy, \qquad x \in \Rm^2,
\ee
where $G$ is the Green's function and $z$ the direction of propagation, (below $\hat G$ denotes the partial Fourier transform of $G$ with respect to the transverse variable $x=(x_1,x_2) \in \Rm^2$ with $|x|=\sqrt{x_1^2+x_2^2}$),
\be \label{Green}
G(x,z)=\frac{e^{i k_0 \sqrt{|x|^2+z^2}}}{4 \pi \sqrt{|x|^2+z^2}}, \qquad \hat G(\xi,z)= \frac{i}{8 \pi^2} \frac{e^{i \sqrt{k_0^2- |\xi|^2} |z|}}{\sqrt{k_0^2- |\xi|^2}},
\ee
and where $U$ is the product of the local field in a given layer and the perturbation of the refractive index due to the heterogeneous medium ($U$ is denoted by $VU$ further, for $V$ the perturbation and $U$ the field). In \cite{chen2020multi}, $f$ is computed using the fast Fourier transform (FFT) algorithm based on the following form of $f$ in the Fourier space:
\be \label{Ff}
f(x)=\frac{1}{(2\pi)^2}\int_{\Rm^2} e^{i \xi \cdot x }\hat G(\xi) \hat U(\xi) d \xi.
\ee
This requires the FFT of $U$, and then an inverse FFT to get $f$.

The problem is actually not as simple as it appears, since $\hat G$ is singular at $|\xi|=k_0$, and it is well known that Fourier-based techniques do not offer high accuracy with singular functions \cite{T-Spectral-00}. The original MLB method avoids the issue of singularity by introducing a cut-off, preventing $|\xi|$ from getting too close to $k_0$. From a physical perspective, this means that only propagating modes are calculated and that the evanescent modes are all discarded. While ignoring the evanescent field would lead to an accurate approximation in a classical far-field scattering problem where propagation occurs in free space, here this unfortunately introduces large errors because the field is always in a near-field regime. Indeed, since the width of the layer has to be small compared to the wavelength for an accurate approximation of the integral equation, scattering is in a near-field regime along $z$ and evanescent modes cannot be neglected. This statement has to be nuanced when the beam enters the scattering medium since it is still well focused, but as it propagates, the beam widens and the transfer from propagating to evanescent modes increases, which in turn enhances backconversion from evanescent to propagating modes. We show theoretically in this work (see Appendix \ref{appev}) that, even in a transverse far-field regime where the norm of $x$ is large compared to the wavelength and $z$ is of the order of the wavelength, the contribution of the evanescent wave is as large as that of the propagating wave provided $\hat U(\xi)$ has a non-zero contribution around $|\xi|=k_0$. This is not surprising since, again, we are in a near-field regime along the axis of propagation.

This previous discussion suggests that evanescent modes have to be included in the simulation for accurate results, which in turn requires careful handling of the calculation of $f$. We propose in this article two methods, both based on the FFT for minimal computational cost. One consists in working with the real-space representation of $f$ and by computing the discrete Fourier transforms (DFT) of $G$ and $U$. The resulting DFT of $G$ is then an effective regularization of the exact $\hat G$ that includes the evanescent modes. Another natural approach consists in ``fixing'' the original MLB formulation based on \fref{Ff} by regularizing directly $\hat G$ and considering a complex-valued wavenumber. We perform a rigorous error analysis that shows that the first method has better accuracy than the second one for a comparable computational cost. Another benefit of the error analysis is to provide us with an optimal way to choose the discretization parameters. 

This strategy solves the evanescent wave problem. A remaining issue is the large computational cost when the layer width is small as more iterations along $z$ have to be performed. We then propose higher-order schemes meant to alleviate the cost based on classical high-order quadrature formulas for numerical integration. We derive a second-order method exploiting the midpoint rule, and a fourth-order method based on Milne quadrature.  

The paper is structured as follows: we define the setting and our two methods for the calculation of \fref{deff} in Section \ref{secnum} and state our error estimates. We present the high-order methods in Section \ref{high}. Our technical results are given in Appendix: in Appendix \ref{deriv}, we recall how to derive the MLB scheme from the Helmholtz equation; we state our two convergence theorems in Appendix \ref{conv}, and detail their proofs in Appendices \ref{proofth1} and \ref{proofth2}. We finally substantiate in Appendix \ref{appev} the claim made in introduction that propagating and evanescent modes have similar amplitudes in a transverse far-field regime.

 \paragraph{Acknowledgments.} This work was funded by grants 2023-336437 and 2024-337798 from the Chan Zuckerberg Initiative DAF, an advised fund of Silicon Valley Community Foundation, and by NSF grant DMS-2404785.

\section{Numerical methods} \label{secnum}

\paragraph{Setting.} The starting point is the three-dimensional Helmholtz equation equipped with Sommerfeld radiation conditions, 
  $$
\Delta U+k^2(x,z) U=S, \qquad \textrm{with} \qquad \partial_r (U-U_{\rm{i}})=ik_0(U-U_{\rm{i}})+O(1/r^2) \qquad \textrm{as} \quad r \to \infty,
$$
where $x=(x_1,x_2) \in \Rm^2$, $r=\sqrt{|x|^2+z^2}$, and the wavenumber reads $k^2(x,z)=(\omega/c)^2 n(x,z)^2$, for $\omega$ the angular frequency, $c$ the speed of light, and $n$ the medium refractive index. The incoming field $U_{\textrm{i}}$ satisfies 
$$
\Delta U_{\rm{i}}+k_0^2U_{\rm{i}}=S.
$$
The constant background index is $n_0$ with $k_0=(\omega/c)n_0$, and let the scattering potential be defined by $V(x,z)=(\omega/c)^2 (n(x,z)^2-n_0^2)$. We suppose that $V$ is supported in a bounded domain $\Omega$. Above, $S$ is a source located outside of $\Omega$. Let the scattered field be defined by $U_{\rm{s}}=U-U_{\rm{i}}$. Calculations given in Appendix \ref{deriv} show, when backscattering is neglected, that $U_{\rm{s}}$ can be approximated by, when $\Delta z >0$,
\bea \nonumber
U_{\rm{s}}(x,z+\Delta z)& \simeq & 
\int_{\Rm^2}  \int_{z}^{z+\Delta z} G(x-x',z+\Delta z-z') \left[V (U_{\rm{i}}+U_{\rm{s}})\right](x',z') dx' dz'\\
&&+P_{\Delta z} U_{\rm{s}}(x,z), \label{startMLB}
\eea
where $P_{\Delta z}$ is the angular propagator defined by
\be \label{angular}P_{\Delta z} \varphi(x,z)=\frac{1}{(2 \pi)^2}\int_{\Rm^2} e^{i \sqrt{k_0^2-|\xi|^2}|\Delta z|} e^{i \xi \cdot x}\hat{\varphi}(\xi,z) d\xi. \ee
Above, the complex square root has positive imaginary part, and we chose the convention $\hat \varphi(\xi)=\int_{\Rm^2} e^{-i \xi \cdot x} \varphi(x)dx$ for the Fourier transform.

Equation \ref{startMLB} is the starting point of MLB and extensions. In the standard MLB, the integral over $z$ is simply approximated by the method of rectangles, resulting in the explicit scheme
\be \label{MLB1}
U_{\rm{s}}(x,z+\Delta z) \simeq \Delta z
\int_{\Rm^2}  G(x-x',\Delta z) \left[V (U_{\rm{i}}+U_{\rm{s}})\right](x',z) dx'+P_{\Delta z} U_{\rm{s}}(x,z),
\ee
which provides us with an iterative formula. 
The key to the method is then to find an efficient and accurate way to compute the convolution with the Green's function.

We propose two methods for this in the next two sections, and perform a rigorous error analysis.
\paragraph{Method 1.} We begin with the $f$ defined in \fref{deff}, in which $U$ now plays the role of the term $V (U_{\rm{i}}+U_{\rm{s}})$. For simplicity, we replace $\Delta z$ by $z>0$ when going from \fref{MLB1} to \fref{deff}, and denote from now on for $z$ fixed $G(x) \equiv G(x,z)$ since $z$ is a fixed parameter. The length $z$ models the width of the layers, and has to satisfy $z \ll \ell=\min(\lambda,\ell_c)$,  where $\lambda = 2\pi /k_0$ is the central wavelength and $\ell_c$ is the typical oscillation length of $V (U_{\rm{i}}+U_{\rm{s}})$, for the expression \fref{MLB1} to yield an accurate approximation of the integral in $z$.



For some length $L>0$, denote by $S_L=[-L/2,L/2]\times [-L/2,L/2]$ the square of side $L$. We choose $L$ sufficiently large so that the support of $U$ is included in $S_L$. We need first to periodize $f$ to compute its Fourier series. For this, we restrict $f$ to $S_L$ and extend it periodically to a function denoted $f_L$. For $x \in S_L$, we have then $f_L(x)=f(x)$, and for $i=1,2$, $f_L(x\pm L e_i)=f_L(x\pm L e_1\pm L e_2)=f(x)$ where $e_1=(1,0)$ and $e_2=(0,1)$. Note that since the function $G$ is even, the periodization we chose does not introduce discontinuities at the boundary of $S_L$ and there is no need for a more involved periodization.

The Fourier series of $f_L$ reads
$$
f_L(x)=\sum_{m,n \in \Zm} \hat f_{m,n} e^{i k_{m,n} \cdot x }, \qquad k_{m,n}=2\pi (m,n)/L,
$$
where the Fourier coefficients are
$$
\hat f_{m,n}= \frac{1}{L^2} \int_{S_L} f(x) e^{-i k_{m,n} \cdot x } dx.
$$
Exploiting the convolution in the definition of $f$ yields
$$
\hat f_{m,n}= L^2  \hat G_{m,n} \hat U_{m,n},
$$
for $\hat G_{m,n}$ and  $\hat U_{m,n}$ the Fourier coefficients of $G$ and $U$, namely
$$
\hat G_{m,n}= \frac{1}{L^2} \int_{S_L} G(x) e^{-i k_{m,n} \cdot x } dx, \qquad \hat U_{m,n}= \frac{1}{L^2} \int_{S_L} U(x) e^{-i k_{m,n} \cdot x } dx.
$$

What can be computed efficiently with the fast Fourier transform algorithm are the DFT of $G$ and $U$, and subsequently an inverse DFT to recover an approximation of $f$. This is formalized as follows. Let $N \in \NN$, chosen to be odd without lack of generality, and denote by $I_N$ the set of integers $I_N=\{-(N-1)/2,\cdots,(N-1)/2\}$ and by $I^c_N=\Zm^2 \backslash I_N \times I_N$ the complement of $I_N \times I_N$ in $\Zm^2$. For $h=L/N$, the set of points $\{x_{p,q}\}_{p,q \in I_N}$, where $x_{p,q}= (p h,qh)$, is a regular discretization of $S_L$. The DFT of $G$ is 
$$
 \hat G^N_{m,n}=\frac{1}{N^2}\sum_{p,q \in I_N} G(x_{p,q})e^{-i k_{m,n} \cdot x_{p,q} }, 
$$
and a similar formula holds for the DFT of $U$ denoted $ \hat U^N_{m,n}$. After an inverse DFT, what we compute numerically is
$$
f_{L,N}(x)= L^2 \sum_{m,n \in I_N} \hat G^N_{m,n} \hat U^N_{m,n}e^{i k_{m,n} \cdot x }.
$$
The goal is to characterize the error between $f$ and $f_{L,N}$ for $x \in S_L$. There are three sources of errors: error in the approximation of the Fourier coefficients of $G$ by its DFT, a similar error for $U$, and the error due to the truncation of the Fourier series of $f_L$. There are two parameters we can control: $L$ and $N$, which determine the size of the domain and the number of discretization points.

We derive an error estimate of $f-f_{L,N}$ in Appendices \ref{conv} and \ref{proofth1}. We express $N$ as $N=N_0 N_\lambda$, where $N_0 \geq 0$, $N_\lambda=[L \lambda^{-1}]$ and $[\cdot]$ denotes integer part. We expect $N_\lambda \gg 1$ in most practical situations.


Note that we do not expect a high accuracy since $G$ is doubly singular: it does not decay fast enough to zero to be integrable, which inevitably introduces errors when periodizing; and it is singular at the origin; the parameter $z$ then acts as a regularization. 

For $u$ a function defined on $S_L$, let
$$
\|u\|^2=\frac{1}{N^2}\sum_{p,q \in I_N}|u(x_{p,q})|^2,
$$
which is an approximation of the average value of $|u|^2$ on $S_L$. Below, the notation $a \lesssim b$ means that there exists a non-dimensional constant $C>0$, independent of the parameters of the problem, such that $a \leq C b$. With $\hat U(k_{m,n})=L^2 \hat U_{m,n}$ and $\hat U^N(k_{m,n})=L^2 \hat U^N_{m,n}$, we obtain the following result, see Appendix \ref{conv},
\be \label{errorM1}
\frac{\|f-f_{L,N}\|}{f_{\textrm{ref}}} \lesssim \frac{1}{N_0^2 N_\lambda}+\frac{\max_{m,n \in I_N} |\hat U(k_{m,n})-\hat U^N(k_{m,n})|}{\|U\|_{L^1(S_L)}},
\ee
where $f_{\textrm{ref}}$ is some reference value for $f$ defined in Appendix \ref{conv} and $\|U\|_{L^1(S_L)}=\int_{S_L} |U(x)|dx$. The first term $1/N_0^2 N_\lambda$ combines the truncation error of the Fourier modes and the DFT error on $G$. It is optimal since all symmetries in $G$ have been exploited to obtain the best possible decay. The second term is the DFT error on $U$, which varies depending on how smooth $U$ is. Note that the estimate is uniform in the parameters $k_0$, $\ell$ and $L$, which is crucial since some constants in non uniform estimates can be large, in particular $k_0L \gg 1$.   

If the perturbation $V$ presents sharp edges, we can only expect a slow decay of the DFT error on $U=V(U_s+ U_{\rm{i}})$ of the form
$$
\frac{\max_{m,n \in I_N} |\hat U(k_{m,n})-\hat U^N(k_{m,n})|}{\|U\|_{L^1(S_L)}} \lesssim \frac{\ell_s}{L N}+\frac{\ell_s^2}{\ell^2 N^2}, 
$$
where $\ell_s$ is the diameter of the transverse support of $V$ in one layer. 
When the edges are smooth, when $U$ is smooth, and when the support of $V$ is large, the accuracy is limited by the oscillations of $U_s+ U_{\rm{i}}$ and $V$, and we have estimates of the form
$$
\frac{\max_{m,n \in I_N} |\hat U(k_{m,n})-\hat U^N(k_{m,n})|}{\|U\|_{L^1(S_L)}} \lesssim \left(\frac{L \ell^{-1}}{N} \right)^q, 
$$
for $q \geq 1$ depending on how smooth $U$ is. When $\lambda$ is of the order of $\ell$, which occurs when the random fluctuations and the incoming field have comparable wavelengths, then $L \ell^{-1} \sim N_\lambda$ and the error above is of order $(1/N_0)^q$, uniformly in the parameters.

We now turn to the second method, which is a modification of the original MLB scheme.

\paragraph{Method 2.} The starting point is to write $f$ as an inverse Fourier transform:

$$
f(x)=\frac{1 }{(2 \pi)^2} \int_{\Rm^2} \hat G(\xi) \hat U(\xi) e^{i \xi \cdot x}d\xi, \qquad \textrm{where} \quad 
\hat G(\xi)= \frac{i}{8 \pi^2}\frac{e^{i \sqrt{k_0^2- |\xi|^2} |z|}}{\sqrt{k_0^2- |\xi|^2}},
$$
and the complex square root has a positive imaginary part. The main issue in the formula above is the singularity of $\hat G(\xi)$ when $|\xi|=k_0$, since a regular grid of wavenumbers $\{k_{m,n}\}_{m,n \in \Zm}$ could have points arbitrarily close to $k_0$ leading to large numerical errors. It is therefore necessary to regularize $\hat G$, and a natural way to proceed is to add an artificial absorption $\eta > 0$ resulting in a complex-valued wavenumber $k_0(\eta)=\sqrt{k_0^2+i\eta^2}$ with positive imaginary part. We then consider the regularized function $f_\eta$, obtained by replacing $G(x)$ by $G_\eta(x)=e^{i k_0(\eta)\sqrt{|x|^2+z^2}}(4 \pi\sqrt{|x|^2+z^2})^{-1}$ in the definition of $f$. We recall that the original MLB method only considers wavenumbers $\xi$ such that $|\xi| < k_0$, and therefore ignores the evanescent modes. We recall that we show in Appendix \ref{appev} that propagating and evanescent modes have a similar amplitude in a transverse far-field regime.

The first step to obtain a computable expression is to truncate the integral over $\xi$: with $K=2 \pi N/L$ and $S_K=[-K/2,K/2]\times [-K/2,K/2]$, let
$$
f_K(x)=\frac{1 }{(2 \pi)^2} \int_{S_K} \hat G_\eta(\xi) \hat U(\xi) e^{i \xi \cdot x}d\xi.
$$
We choose $K > k_0$ in order to account for the evanescent field. Setting $h=L/N$ as in the previous section, let
$$
C_{p,q}=\frac{1}{K^2}\int_{S_K} \hat G_\eta(\xi) \hat U(\xi) e^{i \xi \cdot x_{p,q}}d\xi, \qquad x_{p,q}=h(p,q), \qquad p,q\in I_N,
$$
so that $f_K(x_{p,q})= K^2/(2\pi)^2 C_{p,q}$. The coefficients $C_{-p,-q}$ are the Fourier coefficients of the function $\hat G_\eta(\xi) \hat U(\xi)$.
Let
$$
\hat{U}^N(\xi)=\frac{L^2}{N^2} \sum_{p,q \in I_N} U(x_{p,q}) e^{-i \xi \cdot x_{p,q}},
$$
where $\hat{U}^N_{m,n}=\hat{U}^N(k_{m,n})/L^2$ is the DFT of $U$ which can be computed efficiently. We replace $\hat U$ by its DFT in $C_{p,q}$, giving
$$
\widetilde{C}^N_{p,q}=\frac{1}{K^2}\int_{S_K} \hat G_\eta(\xi) \hat U^N(\xi) e^{i \xi \cdot x_{p,q}}d\xi, \qquad x_{p,q}=h(p,q), \qquad p,q\in I_N,
$$
and we set
$$
\widetilde{f}_{K,N}(x_{p,q})=\frac{K^2}{(2 \pi)^2}\widetilde{C}_{p,q}^N.
$$
The Fourier coefficient $\widetilde{C}_{p,q}^N$ is then approximated by its DFT given by
$$
C_{p,q}^N=\frac{1}{N^2}  \sum_{m,n \in I_N} \hat G_\eta(k_{m,n}) \hat U^N(k_{m,n})e^{i k_{m,n} \cdot x_{p,q} }=\frac{L^2}{N^2}  \sum_{m,n \in I_N} \hat G_\eta(k_{m,n}) \hat U_{m,n}^Ne^{i k_{m,n} \cdot x_{p,q} },
$$
where $k_{m,n}= 2 \pi/h (m,n)$ as in the previous section. Collecting all results, we find finally that the function $f_\eta$ is approximated at $x_{p,q}$ by
$$
f_{K,N}(x_{p,q})=\frac{ K^2}{(2 \pi)^2} C_{p,q}^N= \sum_{m,n \in I_N} \hat G_\eta(k_{m,n}) \hat U^N_{m,n}e^{i k_{m,n} \cdot x_{p,q} },
$$
which can be calculated with an inverse DFT. With this construction, there are four sources of error: the regularization error after the introduction of $k_0(\eta)$, the truncation error, and two DFT errors when computing the DFT approximating the Fourier coefficients of $\hat G_\eta \hat U^N$ and of $\hat U$.

There are three free parameters in this discretization (we recall that the layer width $z$ is fixed): $L$, which controls the grid size $2 \pi /L$ in the wavenumber space, the parameter $K$ (or equivalently $N$ since $K= 2 \pi N/L$), characterizing the maximal wavenumber accounted for, and the regularization parameter $\eta$. The theoretical analysis given in Appendix \ref{conv} shows it is beneficial to choose $\eta=k_0 (k_0 L)^{-\frac{1}{3+\eps}}$, for an $\eps>0$ arbitrarily small, and $N=M N_\lambda$ where $N_\lambda$ is the same as in method 1. The parameter $M$ is such that the function $U$ does not contain wavenumbers greater than $M k_0$ (or equivalently, that the contribution of wavenumbers larger than $M k_0$ is negligible). Theorem B.2 shows it is enough to consider $M \lesssim (\lambda/\ell) N_z$ since larger wavenumbers are exponentially damped, for $N_z$ the number of discretization points in the propagation direction $z$. 

In order to compare with method 1, we set $N_0=M$, so that methods 1 and 2 have the same number of discretization points. We recall that $L$ is chosen so that the support of $U$ is enclosed in $S_L$. Under assumptions on $U$ that are verified in our setting of interest, we derive in Appendix \ref{conv} the estimate
\be \label{errorM2}
\frac{\|f-f_{K,N}\|}{f_{\textrm{ref}}} \lesssim \frac{1}{N_\lambda^{\frac{2}{3+\eps}}}+\frac{\max_{k \in S_K} |\hat U(k)-\hat U^N(k)|}{\|U\|_{L^1(S_L)}},
\ee
where $\eps>0$ is arbitrarily small but not zero. We remark first that the DFT error on $U$ is similar to that of method 1, which is expected. The other errors are overall of order $1/N_\lambda^{2/(3+\eps)}$, which is nearly optimal (we expect the optimal rate to be $1/N_\lambda^{2/3} \log N_\lambda$, which can likely be obtained at the price of a more involved analysis). That rate has to be compared with $1/(N_\lambda M^2)$ of method 1, that is significantly better. In particular, the performance of method 1 vs method 2 increases as $M$ gets larger and $U$ contains higher wavenumbers. Note that the computational cost of method 1 is essentially the same as that of method 2 since it only requires in addition the calculation of the DFT of $G$, which can be performed at the beginning of the iterations once and for all. The conclusion is therefore that method 1 offers a better accuracy than method 2 for a comparable cost. 

\medskip

From now on, we only consider method 1 for calculation of the convolution in \fref{deff}. In the case of the first-order MLB \fref{MLB1}, the approximation error of the $z$ integral is of order $\Delta z (\Delta z \ell^{-1} )$ when $U$ is differentiable with respect to $z$, and therefore, after propagating over a distance $L_z$ with $L_z \Delta z^{-1}$ layers, the error is of order $N_z^{-1}=[z \ell^{-1}]$. This can be improved by deriving high-order versions of MLB. In a globally fourth-order discretization of the $z$ integral, the error is of order $N_z^{-4}$. High-order schemes are derived in the next section. 

The previous discussion applies to the case where $V$ is smooth in the $z$ variable, which occurs for instance when $V$ models smooth random perturbations. When $V$ models an inclusion such as a sphere with a different refractive index than the background, then $V$ presents sharp transitions and the quadrature error in one layer is of order $\Delta z \Delta V \ell_V$, where $\Delta V$ quantifies the jump in refractive index and $\ell_V$ is the length of the boundary of the transverse support of $V$ in the layer. After passing through all layers, the error is $\Delta z \Delta V$ times the surface of the scattering object. In that situation, there is no interest in moving to high-order schemes since the error due to the jump singularity will dominate.

  \section{High-order MLB} \label{high}

  The starting point is expression \fref{startMLB}. Different schemes are obtained depending on the order of the quadrature rule used to approximate the integral with respect to $z$. The simplest one is the first-order MLB derived below. The calculation of the angular propagator $P_{\Delta z}$ is also necessary and is equivalent to computing the convolution \fref{deff} with $G$ replaced by $\partial_z G$. It is natural to use the Fourier representation of the propagator since the kernel $e^{i \sqrt{k_0^2-|\xi|^2}|\Delta z|}$, contrary to that of $G$, does not exhibit the singularity at $k_0$. Now, a closer look shows that the Fourier transform of the scattered field $U_{\textrm{s}}$ to be propagated by $P_{\Delta z}$ and originating from the convolution in expression \fref{MLB1}, is proportional to $\hat G$, which has the singularity. Hence, the evaluation of $P_{\Delta z} U_{\textrm{s}}$ is essentially similar to that of the convolution in \fref{MLB1} since now $\hat U_{\textrm{s}}$ is singular, and we expect a comparable accuracy as in the calculation of \fref{deff}.

  Note that the methods described in the previous section used to estimate the errors can be adapted to the calculation of $P_{\Delta z} U_{\textrm{s}}$, and in addition to the DFT error due to $U_{\textrm{s}}$, method 1 has an error of order $N_z/(M^2 N_\lambda)$, while that of method 2 is now $1/N_\lambda^{3/2}$. Method 1 becomes worse because of the increased singularity around zero in $\partial_z G$, while in the Fourier space, method 2 improves since $e^{i \sqrt{k_0^2-|\xi|^2}|\Delta z|}$ has one more derivative than $\hat G$ with respect to $\xi$. In practical situations, we expect $N_\lambda$ to be significantly larger than $M$ and $N_z$, which makes method 2 more accurate and is our preferred choice here. As already mentioned, the dominant error term in the calculation of $P_{\Delta z} U_{\textrm{s}}$ is due to the DFT error on $U_{\textrm{s}}$, and this quantity is obtained in the previous section with method 1.
      
  We start with a modified version of first-order MLB.
  
\paragraph{First-order MLB.}
The integral with respect to $z$. in \fref{startMLB} is approximated using the rectangle rule, leading to
$$
\Delta z \int_{\Rm^2}  G(x-x',\Delta z ) \left[V (U_{\textrm{i}}+U_{\textrm{s}})\right](x',\Delta z) dx'+O\left(\Delta z \frac{\Delta z}{\ell}\right),
$$
where we recall that $\ell=\min(\lambda,\ell_c)$ and we use the Landau notation $O(\cdot)$. Denoting by $\star$ the convolution in the $x$ variable (evaluated numerically using method 1), and setting $G_z(x)=G(x,z)$, the first order scheme reads for $n \geq 0$,
\bee
U^{n+1}_{\rm{s}}(x)&=& 
\Delta z G_{\Delta z} \star \left[V^n (U^n_{\rm{i}}+U^n_{\textrm{s}})\right]+P_{\Delta z} U^n_{\rm{s}}(x),
\eee
with initialization $U_s^0=0$. Above, we have set $V^n=V(z_n,x)$, $U^n_{\rm{i}}=U_{\rm{i}}(z_n,x)$, where $z_n$ denotes the entry position of the $n$-th layer, and $U^{n+1}_{\rm{s}}(x)$ is an approximation of the scattered field in this $n$-th layer. The term $P_{\Delta z} U^n_{\rm{s}}$ is computed using method 2 without the regularization $\eta$ since it is not necessary here. 


\paragraph{Second-order MLB.} 
We now use a second-order quadrature rule for the integral. We must be careful in doing so since the standard trapezoidal rule involves the end points $0$ and $\Delta z$ and therefore the function $G(x,0)$, which is singular at $|x|=0$. We then use the midpoint rule to approximate the integral, which is of order two and does not require the endpoints. We have then
\begin{align*}
  & \int_{0}^{\Delta z} G(x-x',\Delta z-z') \left[V (U_{\rm{i}}+U_{\rm{s}})\right](x',z'+z) dz'\\
  &=\Delta z \; G\left(x-x',\frac{\Delta z}{2}\right) \left[V (U_{\rm{i}}+U_{\rm{s}})\right]\left(x',z+\frac{\Delta z}{2}\right)+O\left(\Delta z \left(\frac{\Delta z}{\ell}\right)^2\right),
\end{align*}
and we use first order MLB to get $U_{\rm{s}}\left(x',z_n+\frac{\Delta z}{2}\right)$:
\bee
U_{\rm{s}}\left(x,z_n+\frac{\Delta z}{2}\right)&= &
\frac{\Delta z}{2} \; G_{\Delta z/2} \star [V^n (U^n_{\rm{i}}+U^n_{\rm{s}})]+P_{\Delta z/2}U_{\rm{s}}^n+O\left(\Delta z \frac{\Delta z}{\ell}\right).
\eee
The second-order scheme then reads, for $n \geq 0$,
\bee
U^{n+1}_{\rm{s}}&=& 
\Delta z \; G_{\Delta z/2} \star [V^{n+\frac{1}{2}} (U^{n+\frac{1}{2}}_{\rm{i}}+U^{n+\frac{1}{2}}_{\rm{s}})]+P_{\Delta z} U_s^n
\eee
where $U_{\rm{s}}^0=0$ and
\bee
U_{\rm{s}}^{n+\frac{1}{2}}&=&\frac{\Delta z}{2}\; G_{\Delta z/2} \star [V^n (U^n_{\rm{i}}+U^n_{\rm{s}})]+P_{\Delta z/2} U_{\rm{s}}^n, \qquad n \geq 0.
\eee

One iteration of second-order MLB involves two additional convolutions to compute compared to first order MLB, for a total of four convolutions.

\paragraph{Fourth-order MLB.} As in the second-order case, we cannot use a Simpson type formula of order four since it involves the end points. We instead employ Milne quadrature rule, which is a globally fourth order open Newton-Cotes type method to approximate the integral. The Simpson rule would have lead to the standard Runge-Kutta 4 algorithm (RK4), while the Milne rule yields a modified version of RK4. 

We first rewrite the scattered field as 
$$
U_{s}\left(z+\Delta z\right)=U_1\left(z+\Delta z\right)+U_2\left(z+\Delta z\right),
$$
where
\bee
U_1\left(z+\Delta z\right)&=&P_{\Delta z} U_s(z)\\
 U_2\left(z+\Delta z\right)&=&\int_{0}^{\Delta z}\int_{\Rm^2 }G(x-x',\Delta z-z') \left[V (U_{\rm{i}}+U_1+U_2)\right](x',z+z')dz'  dx'.
\eee

The Milne quadrature then yields
\bee
U_2\left(z+\Delta z\right)&=&\frac{2 \Delta z}{3}G_{3 \Delta z/4} \star \left[V (U_{\rm{i}}+U_1+U_2)\right](\Delta z/4)\\
&&-\frac{\Delta z}{3}G_{\Delta z/2} \star \left[V (U_{\rm{i}}+U_1+U_2)\right](\Delta z/2)\\
&&+\frac{2 \Delta z}{3}G_{\Delta z/4} \star \left[V (U_{\rm{i}}+U_1+U_2)\right](3\Delta z/4)+O\left(\Delta z \left(\frac{\Delta z}{\ell}\right)^4\right).
\eee
There are two points to address in the formula above. First, $U_2$ is not known at the locations $\Delta z/4$, $\Delta z/2$ and $3\Delta z/4$. We then use explicit approximations of $U_2$ at these points in the spirit of RK4 that preserve the overall order of accuracy. Second, $U_1$ needs to be computed at $\Delta z/4$, $\Delta z/2$ and $3\Delta z/4$, which is done as in the second-order case by using the propagator $P_{\Delta z}$. 
Let
\bee
H^n(t,X)&=&G_{\Delta z(1-t)} \star \left[V^{n+t} (U^{n+t}_{\rm{i}}+U^{n+t}_1+X)\right] \qquad n \geq 0.
\eee
The fourth order scheme then reads, for $n \geq 0$:
\bee
U_s^{n+1}&=&U_{1}^{n+1}+U_{2}^{n+1}\\
U_{1}^{n+t} &=& P_{t \Delta z} U_\textrm{s}^{n}, \qquad U_{1}^{0}=0\\
U_{2}^{n+1} &=&\frac{\Delta z}{3} (2 h^n_1-h^n_2+2h^n_3),
\eee
where
\bee
 h_1^n&=&H^n\left(\frac{1}{4},X^n_1\right), \qquad X^n_1=\frac{\Delta z}{4} G_{\Delta z/4} \star \left[V^{n} (U^{n}_{\rm{i}}+U^{n}_1)\right] \\ 
 h^n_2&=&H^n\left(\frac{1}{2},X_2^n\right), \qquad X^n_2=\frac{\Delta z}{4} G_{\Delta z/4} \star \left[V^{n+\frac{1}{4}} (U^{n+\frac{1}{4}}_{\rm{i}}+U^{n+\frac{1}{4}}_1+X^n_1)\right]\\
 h^n_3&=&H^n\left(\frac{3}{4},X_3^n\right), \qquad X^n_3=\frac{\Delta z}{4} G_{\Delta z/4} \star \left[V^{n+\frac{1}{2}} (U^{n+\frac{1}{2}}_{\rm{i}}+U^{n+\frac{1}{2}}_1+X^n_2)\right].\\
 \eee

 Computing $X_1^n, X_2^n,X_3^n$ requires 5 convolutions, and $h_1^n, h_2^n, h^n_3$ requires 4 convolution, so that computing $U_{2}^{n+1}$ involves 9 convolutions. In total, $U_{\textrm{s}}^{n+1}$ requires then 10 convolutions,
 that is 8 more than the first-order MLB. 

 The higher-order schemes, whether second- or fourth-order, require $V (U_{\rm{i}}+U_{s})$ be smooth since the quadrature error depends on the $z$ derivative of $G \star V (U_{\rm{i}}+U_{s})$ (third order derivative for second-order MLB and fifth-order for fourth-order MLB). In particular, if the $n$-th $z$ derivative of $G \star V (U_{\rm{i}}+U_{s})$ behaves like $\Delta z^{-n}$, higher-order schemes will not be useful. We do not have this issue here since, while $G$ is not itself regular with respect to $z$, a $z$ derivative can be exchanged for an $x$ derivative due to the particular form of $G$. More precisely, we write, with the shorthand $U=V (U_{\rm{i}}+U_{s})$,
 \begin{align*} \nonumber 
   \int_{\Rm^2} &G(x',z-z') U(x-x',z') dx'\\
&=\int_{\Rm^2} G(x',z-z') (U(x-x',z')-U(x,z')) dx'+U(x,z')\int_{\Rm^2} G(x',z-z') dx' ,
\end{align*}
and realize that, when $z-z'>0$,
$$
\int_{\Rm^2} G(x',z-z') dx'=\frac{i}{8 \pi^2 k_0} e^{i k_0 (z-z')} 
$$
which is smooth with respect to $z$. Then, since 
$$
U(x-x',z')-U(x,z')\simeq -x' \cdot \nabla_x U(x,z'),
$$
the $x'$ before the gradient on the r.h.s.  compensates for the singularity arising when differentiating $G$ with respect to $z$. Hence, a $z$ derivative on $G$ can be exchanged for an $x$ derivative on $U$. The same idea applies to high-order derivatives.

 \section{Conclusion} We proposed in this work two remedies for the evanescent mode issue of the original MLB. We showed with a rigorous error analysis that the one based on a direct computation of the DFT of the Green's function provides us with a natural regularization and yields better error estimates than the one based on regularizing by adding an artificial absorption. The mathematical analysis also provides us with an optimal way to choose the discretization parameters.

 We have moreover derived high-order numerical schemes based on high-order quadrature rules for numerical integration. These methods theoretically offer a lower computation cost for a given accuracy than lower-order schemes, and they are applicable when the background perturbations are smooth, as for instance in wave propagation in random smooth media. A numerical validation of these algorithms will be the object of a future work.

 An aspect that is not covered in this article is the stability of the numerical scheme, namely what is the condition of the slab width $\Delta z$ for the iterates of the field to remain bounded. Runge-Kutta type methods for ordinary differential equations are conditionally stable, i.e. $\Delta z$ has be small enough for the solution to not blow up, but in general the stability criteria are not too strict. We expect the same to hold here since our schemes are variations of Runge-Kutta methods.

 \begin{appendix}

\section{Derivation of the MLB scheme} \label{deriv}
The starting point is the three-dimensional Helmholtz equation and we use notations introduced in Section \ref{secnum}. Since the Helmholtz equation can be recast as 
$$
\Delta (U-U_{\rm{i}})+k^2_0 (U-U_{\rm{i}})+VU=0,
$$
we can express $U$ with the integral equation
$$
U(x,z)=U_{\rm{i}}(x,z)+ \int_{\Rm^2} \int_{\Rm} G(x-x,z-z') (V U)(x',z') dx' dz',
$$
where $G$ is the Green's function defined in \fref{Green}. It is more convenient to write the integral equation only in terms of the scattered field $U_{\rm{s}}=U-U_{\rm{i}}$, i.e.
$$
U_{s}(x,z)= \int_{\Rm^2} \int_\Rm G(x-x',z-z') \left[V (U_{\rm{i}}+U_{s})\right](x',z') dx' dz'.
$$
We suppose the incoming wave is propagating along the positive $z$ axis and then separate the forward wave and backscattering as follows:
$$
U_{s}(x,z)=\int_{\Rm^2} \int_{-\infty}^z G(x-x',z-z') \left[V (U_{\rm{i}}+U_{s})\right](x',z') dx' dz'+B_{sc}(x,z),
$$
where
$$
B_{sc}(x,z)=\int_{z}^{\infty} G(x-x',z-z') \left[V (U_{\rm{i}}+U_{s})\right](x',z') dx' dz'.
$$
The goal now is to write the field at  $z+\Delta z$ in terms of the one at $z$. We begin with
\bee
U_{s}(x,z+\Delta z)&=& \int_{\Rm^2} \int_{-\infty}^{z+\Delta z} G(x-x',z+\Delta z-z') \left[V (U_{\rm{i}}+U_{s})\right](x',z') dx' dz'\\
&&+B_{sc}(x,z+\Delta z)\\
&=&  \int_{\Rm^2} \int_{z}^{z+\Delta z} G(x-x',z+\Delta z-z') \left[V (U_{\rm{i}}+U_{s})\right](x',z') dx_\perp' dz'\\
&&+\int_{\Rm^2} \int_{-\infty}^{z} G(x-x',z+\Delta z-z') \left[V (U_{\rm{i}}+U_{s})\right](x',z') dx' dz'\\
&&+B_{sc}(x,z+\Delta z).
\eee
In order to get a recursive formula, we use the angular propagator $P_{\Delta z}$ defined in \fref{angular} and observe that
$$
P_{\Delta z} G(x-x',z-z')=G(x-x',|z-z'|+|\Delta z|).
$$
When $z>z'$ and $\Delta z >0$, this gives $G(x-x',z-z'+\Delta z)=P_{\Delta z} G(x-x',z-z')$, and it follows that
\bee
U_{s}(x,z+\Delta z)&=& 
\int_{\Rm^2}  \int_{z}^{z+\Delta z} G(x-x',z+\Delta z-z') \left[V (U_{\rm{i}}+U_{s})\right](x',z') dx' dz'\\
&&+P_{\Delta z} U_{\rm{s}}(x,z)-P_{\Delta z} B_{sc}(x,z)+ B_{sc}(x,z+\Delta z).
\eee
At this stage there is no approximation and the integral equation above is exact. A computationally tractable equation is obtained by exploiting the directionality of the incoming field and by neglecting backscattering, giving the expression
\bee 
U_{s}(x,z+\Delta z)& \simeq & 
\int_{\Rm^2}  \int_{z}^{z+\Delta z} G(x-x',z+\Delta z-z') \left[V (U_{\rm{i}}+U_{s})\right](x',z') dx' dz'\\
&&+P_{\Delta z} U_{s}(x,z), 
\eee
which is the starting point of MLB and extensions. 

\section{Convergence results} \label{conv}
We use the notations of Section \ref{secnum}. Consider the non-dimensional constants
\bee
c_1&=& (k_0z)^{-1}+k_0 L\\
c_2&=&L k_0^{-1}z^{-2}+k_0^2 L^2.
\eee
The next theorem, proved in Section \ref{proofth1}, concerns the convergence of method 1. We purposedly keep track of the dependency of various constants on important parameters of the problem such as $k_0$, $L$ and $z$ since these constants are large in relevant regimes.

\begin{theorem} \label{th1} We have the estimate
  $$
\|f-f_{L,N}\| \lesssim E_{\textrm{trunc+DFT-G}}+E_{\textrm{DFT-U}}.
$$
where, for all $\eps \in (0,1)$,
\bee
E_{\textrm{trunc+DFT-G}} &=&k_0\big(L^2\|U\|+L\|U\|_{L^2(S_L)}\big) \left( \frac{c_1}{N^2}+\frac{c_1^{1-\eps}c_2^\eps}{N^{2(1+\eps)}}\right)\\
E_{\textrm{DFT-U}}&=& L \|G\|_{L^2(S_L)} \max_{m,n \in I_N} |\hat U_{m,n}-\hat U^N_{m,n}|.
\eee
\end{theorem}

The term $E_{\textrm{trunc+DFT-G}}$ combines the error due to the truncation of the Fourier series of $f$ and the error due to  the approximation of the Fourier coefficients of $G$ by the DFT. The second term, $E_{\textrm{DFT-U}}$, is the error due to the approximation of the Fourier coefficients of $U$ by the DFT. We observe that $E_{\textrm{trunc+DFT-G}}$ decays quadratically in $N$, and therefore the error decreases as can be expected as the grid get finer. The quadratic decay is optimal as all symmetries of $G$ were exploited in the estimate.

The constants $c_1$ and $c_2$ are actually large since $k_0 L \gg 1$ and $z \ll \lambda$, as we expect oscillations over multiple wavelengths across the transverse domain. The parameters $N$ and $L$ have then to be chosen accordingly. Let first $N_z=[\ell z^{-1}]$ and $N_\lambda=[L \lambda^{-1}]$, where $[\cdot]$ denotes integer part. We have $N_z \gg 1$ and $N_\lambda \gg 1$. We have already assumed that $L$ is greater than the diameter of the support of $U$. Supposing for instance that $N_z=10$ for reasonable accuracy, we expect in practical situations that $N_\lambda \geq 10$, namely that the transverse support of $U$ is at least ten wavelengths wide, so that $(k_0 z)^{-1} \lesssim k_0 L$. The dominating term in $c_1$ is then $k_0L$.
%

Hence, $c_1 \lesssim N_\lambda$, and owing to the fact that $Lk_0^{-1} z^{-2}= Lk_0 (k_0 z)^{-2}$, we have  $c_2 \lesssim N_\lambda^3$. Setting $N=N_\lambda N_0$ and defining as reference
$$
f_{\textrm{ref}}=\max\Big(k_0\big(L^2\|U\|+L\|U\|_{L^2(S_L)}\big), L^{-1} (\|G\|_{L^2(S_L)}+\|\hat{G}\|_{L^2(S_K)} )\|U\|_{L^1(S_L)}\Big),
$$
we obtain \fref{errorM1}. This choice for the reference is  motivated by two facts: (i) the terms appearing in $f_{\textrm{ref}}$ are precisely the ones obtained in the estimates of Theorems 1 and 2 below, and (ii) the square of the second term is a direct estimate of the average value of $|f|^2$ over $S_L$.

In the scenario where $N_z \gg N_\lambda$, then the rate $1/(N_0^2N_\lambda)$ has to be changed to $ N_z/(N_0^2N_\lambda^2)$. This occurs when the support of $U$ is of the order of the wavelength or smaller, and models a localized scatterer. This is not the practical context we have in mind here since we are mostly interested in extended scattering media.

Regarding method 2, we will need the following quantity to state our main result: let
\begin{align} \nonumber 
  &k_U=\max \left(\frac{\max_{j=1,2} \| \partial_{\xi_j} \hat U^N\|_{L^1(S_K)}}{\max_{\xi \in S_K}  |\hat{U}^N(\xi)|},\frac{\max_{i,j=1,2} \| \partial^2_{\xi_i \xi_j} \hat U^N\|_{L^1(S_K)}}{L \max_{\xi \in S_K}  |\hat{U}^N(\xi)|}\right.\\\label{defKU}
  &\hspace{8cm}\left.\frac{\max_{i,j=1,2} \| \partial^2_{\xi_i^2 \xi_j} \hat U^N\|_{L^1(S_K)}}{L^2 \max_{\xi \in S_K}  |\hat{U}^N(\xi)|}\right),
\end{align}
which is a length quantifying how steep are the derivatives of $\hat U^N(\xi)$ after integration over $S_K$. 
We make the following assumptions that are not necessary but simplify expressions in the next theorem: $(k_0 L)^{-1} \lesssim \eta_0$ where $\eta^2=k_0^2 \eta_0$, and $\log(k_0 z) \leq \eta_0^{-1/2}$.
We have then the

\begin{theorem} \label{th2} The following estimate holds:
  \begin{align*}
    &\|f-f_{K,N}\| \lesssim \calE_{\textrm{reg}}+\calE_{\textrm{trunc}}+\calE_{\textrm{DFT}}+\calE_{\textrm{DFT-U}},
  \end{align*}
  where, for any $\eps \in (0,1)$,
  \bee
  \calE_{\textrm{reg}}  &=& \|U\|_{L^1(S_L)} k_0 \eta_0\\
  \calE_{\textrm{trunc}}  &=& \max_{|\xi| \geq K } |\hat{U}(\xi)| z^{-1} e^{-z \sqrt{K^2-k_0^2}}\\
   \calE_{\textrm{DFT}}&=& \|U\|_{L^1(S_L)} L^{-1} ( \alpha_U+\eta_0^{-1/2-\eps})\\
   \calE_{\textrm{DFT-U}}&=& L^{-1}\| \hat G\|_{L^2(S_K)} \max_{\xi \in S_K} |\hat U(\xi)-\hat U^N(\xi)|
   \eee
   and where $\alpha_U = k_U k_0^{-1} \eta_0^{-1/2}$.
 \end{theorem}

 As expected, the regularization error  $\calE_{\textrm{reg}}$ goes to zero as $\eta_0 \to 0$, and  $\calE_{\textrm{trunc}} \to 0$ as the truncation parameter $K \to + \infty$. The parameter controlling the DFT error is $L$, as the cell size in wavenumbers is $2 \pi /L$. Observe the term $\eta_0^{-1/2-\eps}$ which is reminiscent from the singular nature of $\hat G$. These estimates are nearly optimal, we believe the term $\eta_0^{-1/2-\eps}$ can be transformed into  $\eta_0^{-1/2} |\log \eta_0|$  at the price of more involved calculations. 

 We use now the previous theorem to set the parameters $\eta_0$, $K$ and $L$. Ignore for the moment $\alpha_U$ and set it to zero. We then choose $\eta_0$ such that $\calE_{\textrm{reg}}$ and $\calE_{\textrm{DFT}}$ are about the same order, that is we set $k_0 L \eta_0= \eta_0^{-1/2-\eps}$, which is $\eta_0 =(k_0 L)^{-\frac{2}{3+2\eps}} \ll 1$. Note that with such a $\eta_0$, the assumptions  $(k_0 L)^{-1} \lesssim \eta_0$ and $\log(k_0 z) \leq \eta_0^{-1/2}$ are satisfied.

 Concerning $\calE_{\textrm{trunc}}$, we assume that $\hat U$ does not have contributions for wavenumbers greater than some $ M  k_0$ (or equivalently that contributions for $|\xi| \geq M k_0$ are negligible), for some non-dimensional number $M$. We then set $K=M k_0$, and writing $K$ as $K= 2 \pi N/L$, this amounts to choose $N$ of the order of $M k_0 L$, that is $M N_\lambda$ with previous notation. In that case, $\calE_{\textrm{trunc}}$ vanishes or is negligible compared to the other errors. We then obtain the estimate \fref{errorM2}. This shows that the error $\|f-f_{K,N}\|$ is at best as \fref{errorM2} when $\alpha_U \neq 0$. Note that the exponential term in $\calE_{\textrm{trunc}}$ damps wavenumbers that are of the order of $\Delta z^{-1}$, so that it is enough to consider values of $M$ that are less than $(\lambda/\ell)N_z$ since larger wavenumbers are exponentially damped.  


 We can actually estimate $\alpha_U$ in our cases of interest. We first replace $\hat U^N$ by $\hat U$ since the former is an approximation of the latter and we expect the associated $\alpha_U$ to be comparable. We treat $U_{\textrm{i}}$ and $U_s$ separately. 
  
 In practice $U_{\textrm{i}}$ is often a plane wave, and as a consequence the Fourier transform of $V U_{\textrm{i}}$ is obtained by shifting that of $V$. When $V$ models an heterogeneous medium, it is usually obtained numerically as an inverse DFT, that is
$$
V(x)=\varphi(x/L_\varphi)\sum_{m,n \in I_N} R(k_{m,n}) e^{i k_{m,n} \cdot x },
$$
for some enveloppe function $\varphi$ defining the support of $V$ with diameter $L_\varphi$ and a function $R$ characterizing its the spatial frequency content. Its Fourier transform reads
$$
\hat V(\xi) = L_\varphi^2 \sum_{m,n \in I_N} R(k_{m,n}) \hat \varphi (L_\varphi(\xi-k_{m,n})).
$$
By construction $L_\varphi \leq L$, and a short calculation shows that $k_U$ is of the order of $L^{-1}_\varphi$. In that case, $\alpha_U$ is of the order of $L^{-1}_\varphi k_0^{-1} \eta_0^{-1}$. Supposing the incoming field oscillates multiple times over the support of $V$, we have $(k_0 L_\varphi)^{-1} \lesssim 1$ and $\alpha_U \lesssim n_0^{-1/2}$. This means that \fref{errorM2} is actually sharp.

The situation is slightly different for the term $V U_s$. Let us consider expression \fref{startMLB} and ignore the angular propagator part since it is more regular than the first term. We then write $U_s= G_z \star W$ for some $W$. The leading part in $\widehat{V U_{s}}$ is then given by the convolution of $\hat V$ and $\hat G \hat W$. 
We have $\| \partial_{\xi_j} \hat V \star \hat U_s \|_{L^1(S_K)} \sim \| \partial_{\xi_j} \hat V \|_{L^1(S_K)} \| \hat G_z  \hat W\|_{L^1(S_K)} \sim L_\varphi \| \hat G_z  \hat W\|_{L^1(S_K)}$. On the other hand, since $\max_{\xi \in S_K} |\hat V(\xi)| \sim L_\varphi^2$, we have $\max_{\xi \in S_K} |\hat V \star (\hat G_z W)| \sim L_\varphi^2 \| \hat G_z  \hat W\|_{L^1(S_K)}$. There are similar relations for higher-order derivatives, and in the end $k_U/k_0$ is of the order of $(k_0 L_\varphi)^{-1}\lesssim 1$. We have therefore $\alpha_U \lesssim n_0^{-1/2}$ as above. 

We proceed now to the proof of Theorem \ref{th1}.

 \section{Proof of Theorem \ref{th1}} \label{proofth1}

 The first step is to quantify the decay of the Fourier coefficients of $G$.

 \paragraph{Step 1: Estimates on $\hat G_{m,n}$.} We have the following lemma:

 \begin{lemma} \label{lemG} Let
   \bee
   C_0&=&L^{-1}, \qquad C_1= k_0+L^{-1}+z^{-1}+k_0^2 L\\
   C_2&=&k_0+L^{-1}+L z^{-2}+k_0^3 L^2.
   \eee
   Then, for all $m,n \in \Zm$,
   $$
   |\hat G_{m,n}| \lesssim \Big(C^{-1}_0+C_1^{-1}(m^2+n^2)+C_2^{-1}|n|^{3/2}|m|^{3/2}\Big)^{-1}.
   $$
 \end{lemma}
 The estimate above is optimal since we have exploited all symmetries in $G$ and there are boundary terms left that cannot be cancelled. \medskip
 
\begin{proof}
We recall that
$$
G(x)=\frac{e^{i k_0 \sqrt{|x|^2+z^2}}}{4 \pi \sqrt{|x|^2+z^2}}, \qquad \hat G_{m,n}=\frac{1}{L^2} \int_{S_L} G(x) e^{-i k_{m,n} \cdot x } dx.
$$
Denoting by $B_L$ the ball of radius $L$ centered at zero, we have first the following trivial bound
\bea \label{estG1}
|\hat G_{m,n}| &\leq& \frac{1}{L^2} \int_{B_L} |G(x)|dx \leq \frac{2 \pi }{L^2} \int_{0}^L \frac{r dr}{4 \pi r} \lesssim  L^{-1}.
\eea
We estimate now the partial derivatives of $G$. Let $r=\sqrt{|x|^2+z^2}=\sqrt{x_1^2+x_2^2+z^2}$. Then
$$
\partial_{x_1}G(x)=ik_0 x_1\frac{e^{i k_0 r}}{4 \pi r^2}-x_1\frac{e^{i k_0 r}}{4 \pi r^3}=x_1\big(ik_0-1/r\big)\frac{e^{i k_0 r}}{4 \pi r^2}
$$
and
\bee
\partial^2_{x_1^2}G(x)&=&\big(ik_0-1/r\big)\frac{e^{i k_0 r}}{4 \pi r^2}+x_1^2\frac{e^{i k_0 r}}{4 \pi r^5}+x_1\big(ik_0-1/r\big)\Big[ik_0 x_1\frac{e^{i k_0 r}}{4 \pi r^3}-x_1 \frac{e^{i k_0 r}}{2 \pi r^4}\Big]\\
&=&\big(ik_0-1/r\big)\frac{e^{i k_0 r}}{4 \pi r^2}+x_1^2\big(ik_0-1/r\big)^2\frac{e^{i k_0 r}}{4 \pi r^3}+x_1^2\frac{e^{i k_0 r}}{4 \pi r^5}.
\eee
Hence,
$$
|\partial_{x_1}G(x)| \lesssim k_0/r+1/r^2,\qquad |\partial^2_{x_1}G(x)| \lesssim r^{-3} +k_0^2/r,
$$
and similar calculations give
\be \label{dG2}
|\partial^2_{x_1 x_2}G(x)| \lesssim r^{-3} +k_0^2/r, \qquad |\partial^3_{x^2_1 x_2}G(x)|+|\partial^3_{x_1 x^2_2}G(x)| \lesssim r^{-4} +k_0^3/r.
\ee
Using the previous estimates, we can now control $\hat G_{m,n}$. We have first, after an integration by parts, since $G$ is even, 
$$
\hat G_{m,n}=\frac{1}{2 i \pi  L m} \int_{S_L} \partial_{x_1}G(x) e^{-i k_{m,n} \cdot x } dx.
$$
 Another integration by parts gives
 $$
 \hat G_{m,n}=\frac{1}{(2 \pi)^2 m^2}B - \frac{1}{(2 \pi)^2 m^2} \int_{S_L} \partial^2_{x_1}G(x) e^{-i k_{m,n} \cdot x } dx
 $$
 where the boundary term $B$ can be controlled by 
 $$
 |B| \lesssim k_0+L^{-1}.
 $$
 Hence,
 \be \label{estG2}
  |\hat G_{m,n}| \lesssim |m|^{-2}\big(k_0+L^{-1}+z^{-1}+k_0^2 L\big)=\frac{C_1}{m^2}.
  \ee
  The same calculation as above gives $|\hat G_{m,n}| \lesssim n^{-2}C_1 $. We treat now the mixed derivatives. We use for this the fact that $\partial_{x_1} G(x)$ is even with respect to $x_2$ and obtain
  $$
\hat G_{m,n}=\frac{1}{(2 i \pi )^2 m n} \int_{S_L} \partial^2_{x_1 x_2}G(x) e^{-i k_{m,n} \cdot x } dx.
$$
An integration by parts with respect to $x_1$ gives
 \be  \label{mixed}
 \hat G_{m,n}=\frac{1}{i(2 \pi)^3 m^2 n}B' +\frac{L}{i(2 \pi)^3 m^2 n} \int_{S_L} \partial^3_{x_1^2 x_2}G(x) e^{-i k_{m,n} \cdot x } dx
 \ee
 where the boundary term $B'$ can be controlled by, using \fref{dG2},
 $$
 |B'| \lesssim k_0^2 L+L^{-1}. 
 $$
 Using \fref{dG2} again, the second term in \fref{mixed}  is bounded by
 \bee
 \frac{L}{(2 \pi)^3 m^2 |n|}  \left(k_0^3 L+ \int_0^L \frac{r dr}{(r^2+z^2)^2} \right) &\leq& \frac{L}{(2 \pi)^3 m^2 |n|}  \left(k_0^3 L+ z^{-2}\int_0^\infty \frac{r dr}{(r^2+1)^2} \right)\\
 &\lesssim& \frac{1}{m^2 |n|}  \left(k_0^3 L^2+ L z^{-2}\right).
 \eee
An integration by parts with respect to $x_2$ instead yield a similar result with $m^2 n$ replaced by $m n^2$. Collecting the latter estimate and the one on $B'$, gives, using that $ 2k_0^2 L \leq k_0+ k_0^3L^2$,
  $$
  (m^2 |n|+n^2 |m|)|\hat G_{m,n}| \lesssim  \big(k_0+L^{-1}+Lz^{-2}+k_0^3 L^2\big)=C_2.$$
Since $m^2 |n|+n^2 |m|=|m||n| (|m|+|n|) \geq |m|^{3/2}|n|^{3/2}$, this ends the proof together with \fref{estG1} and \fref{estG2}.
\end{proof}

\medskip

We continue with the truncation error

 \paragraph{Step 2: truncation error.} Let $\widetilde{f}_{L,N}(x)$ be the function obtained by truncating the Fourier series of $f_L$:
$$
\widetilde{f}_{L,N}(x)=\sum_{m,n \in I_N} \hat{f}_{m,n} e^{i k_{m,n} \cdot x }.
$$
Using the fact that, for $m,n \in \Zm$,
$$
\sum_{p,q \in I_N} e^{i k_{m,n} \cdot x_{p,q}}=\frac{\sin(\pi m)\sin( \pi n)}{\sin(\pi m/N)\sin( \pi n/N)}=N^2 \delta (\textrm{mod}(m,N))\delta (\textrm{mod}(n,N)),
$$
where $\delta$ is the Dirac measure, we find
\bee
\|f_L-\widetilde{f}_{L,N}\|^2&=&\frac{1}{N^2}\sum_{p,q \in I_N}|(f_L-\widetilde{f}_{L,N})(x_{p,q})|^2\\
&=&\frac{1}{N^2}\sum_{p,q \in I_N}\sum_{(m,n),(m',n') \in I_N^c} e^{i (k_{m,n}-k_{m',n'}) \cdot x_{p,q}}\hat{f}_{m,n}\hat{f}^*_{m',n'}\\
&=&\sum_{(m,n) \in I_N^c}\sum_{(\ell,\ell') \in J_N(m,n)} \hat{f}_{m,n}\hat{f}^*_{m+ \ell N,n+\ell' N}
\eee
where $ J_N(m,n)=\{ \ell,\ell' \in \Zm: (m + \ell N,n+\ell' N) \in I^c_N\}$ and $I_N^c=\Zm^2 \backslash I_N \times I_N$. Writing $m=m'+ p N$, $n=n'+ p' N$, with $m',n' \in I_N$ and $(p,p') \in \Zm^2_*$, where $\Zm^2_*=\Zm^2 \backslash (0,0)$, the sum above can be recast as
$$
\sum_{m',n' \in I_N}\sum_{(p,p') \in \Zm^2_*} \sum_{(\ell+p,\ell'+p') \in \Zm^2_*}\hat{f}_{m'+pN,n+p'N}\hat{f}^*_{m'+ (\ell+p) N,n'+(\ell'+p') N}.
$$
Using the Cauchy-Schwarz inequality, this can be bounded by
$$
T=\left(\sum_{(p,p') \in \Zm^2_*} \left(\sum_{m',n' \in I_N} |\hat{f}_{m'+pN,n+p'N}|^2\right)^{1/2}\right)^2.
$$
The sum over $(p,p')$ can be split into 3 pieces: (i) $(p,p') \in \Zm_*\times \Zm_*=(\Zm_*)^2$, (ii) $(p=0,p')$ for $p'\in \Zm_*$, and (iii) $(p,p'=0)$ for $p\in \Zm_*$. We denote by $T_i$, $i=1,2,3$ the corresponding terms and start with $T_1$. Let $w(p,p')>0$ be such that $\sum_{(p,p') \in (\Zm_*)^2} (w(p,p'))^{-1} <\infty$.
Recalling that $
\hat f_{m,n}= L^2  \hat G_{m,n} \hat U_{m,n}$, the Cauchy-Schwarz inequality applied once more gives
\bee
T &\leq & L^4 M_N \sum_{(p,p') \in (\Zm_*)^2} (w(p,p'))^{-1} \sum_{(p,p') \in (\Zm_*)^2} \sum_{m',n' \in I_N} |\hat{U}_{m'+pN,n+p'N}|^2\lesssim L^4 M_N \sum_{m,n \in \Zm} |\hat{U}_{m,n}|^2
\eee
where
$$
M_N=\max_{(\ell,\ell') \in (\Zm_*)^2}  \max_{m,n \in I_N}  w(p,p') |\hat{G}_{m+ \ell N,n+\ell' N}|^2.
$$
According to Lemma \ref{lemG}, we have the estimate $ |\hat{G}_{m,n}| \lesssim C_1|m n|^{-1}$. Setting $w(p,p')=(1+|p|)^{3/2}(1+|p'|)^{3/2}$, it follows that
$M_N \lesssim (C_1N^{-2})^2$.
Since Parseval's inequality yields moreover
$$
\sum_{m,n \in \Zm}|\hat U_{m,n}|^2=\frac{1}{L^2} \int_{S_L} |U(x)|^2 dx=\frac{1}{L^2} \|U\|^2_{L^2(S_L)},
$$
we find
$$
T \leq L^2(C_1 N^{-2})^2\|U\|^2_{L^2(S_L)}.
$$
The terms $T_2$ and $T_3$ can be treated in a similar way, and the final estimate is
\be \label{trunc}
\|f_L-\widetilde{f}_{L,N}\| \lesssim  L \|U\|_{L^2(S_L)} C_1/N^2.
\ee


This gives the truncation error. The next step is to use once more the discrete Parseval equality to arrive at
$$
\|\widetilde{f}_{L,N}-f_{L,N}\|^2=L^4 \sum_{m,n \in I_N}|\hat G_{m,n}\hat U_{m,n}-\hat G_{m,n}^N\hat U^N_{m,n}|^2.
$$
The triangle inequality then gives
\bee
 L^{-2}\|\widetilde{f}_{L,N}-f_{L,N}\| &\leq& \left( \sum_{m,n \in I_N}|\hat G_{m,n}(\hat U_{m,n}-\hat U^N_{m,n})|^2\right)^{1/2}\\
&&+ \left( \sum_{m,n \in I_N}|\hat U^N_{m,n}(\hat G_{m,n}-\hat G^N_{m,n})|^2\right)^{1/2}.
\eee
Parserval's inequality yields
$$
\sum_{m,n \in I_N}|\hat G_{m,n}|^2 \leq \sum_{m,n \in \Zm}|\hat G_{m,n}|^2=\frac{1}{L^2} \int_{S_L} |G(x)|^2 dx=\frac{1}{L^2} \|G\|^2_{L^2(S_L)}
$$
and we have as well
$$
\sum_{m,n \in I_N}|\hat U^N_{m,n}|^2 =\frac{1}{N^2}\sum_{p,q \in I_N} |U(x_{p,q})|^2=\|U\|^2.
$$
Hence
\be \label{errorFT}
\|\widetilde{f}_{L,N}-f_{L,N}\| \lesssim  L^2 \left( L^{-1} \|G\|_{L^2(S_L)} \max_{m,n \in I_N} |\hat U_{m,n}-\hat U^N_{m,n}|+\|U\| \max_{m,n \in I_N} |\hat G_{m,n}-\hat G^N_{m,n}| \right).
\ee

 \paragraph{Step 3: the DFT error $\hat G_{m,n}-\hat G_{m,n}^N$.} 
We use the following aliasing formula (adapted to 2D), see \cite{gasquet2013fourier} Section 8.4, where the sum runs over $p$ and $q$ in $\Zm$,
   $$
 \hat G^N_{m,n}-\hat G_{m,n}=\sum_{(p,q) \neq (0,0)} \hat G_{m+pN,n+qN}, \qquad m,n \in I_N.
 $$
 We break the sum into 3 terms:
 \be \label{3T}
 \hat G^N_{m,n}-\hat G_{m,n}=\sum_{q \neq 0} \hat G_{m,n+qN}+\sum_{p \neq 0} \hat G_{m+pN,n}+\sum_{p \neq 0}\sum_{q \neq 0} \hat G_{m+pN,n+qN}.
 \ee
 Using Lemma \ref{lemG}, the first term is bounded by
 $$
 C_1\sum_{q \neq 0} \frac{1}{(n+qN)^2}=\frac{C_1}{N^2}\sum_{q \neq 0} \frac{1}{(\frac{n}{N}+q)^2} \lesssim \frac{C_1}{N^2},
 $$
 where we used that $|n/N|<1/2$ since $n \in I_N$. We obtain the same result for the second term in \fref{3T}. With Lemma \ref{lemG}, we control the third term in \fref{3T} by, for any $\eps \in (0,1)$,
 $$
  C^{1-\eps}_1 C^\eps_2\sum_{p \neq 0}\sum_{q \neq 0} \frac{1}{(n+qN)^{1+\eps}(m+pN)^{1+\eps}} \lesssim  \frac{C^{1-\eps}_1 C^\eps_2}{N^{2(1+\eps)}}.
  $$
  We introduced the parameter $\eps$ in order to minimize the power of the (large) constant $C_2$. 
 As a conclusion, we have, for all $m,n \in I_N$:
 \be \label{GDFT}
 |\hat G_{m,n}-\hat G_{m,n}^N| \lesssim \frac{C_1}{N^2}+\frac{C^{1-\eps}_1 C^\eps_2}{N^{2(1+\eps)}}, \qquad \forall \eps \in (0,1).
 \ee
 \paragraph{Step 4: Conclusion.} The triangle inequality gives
 $$
\|f_L-f_{L,N}\| \leq \|f_{L}-\widetilde{f}_{L,N}\| +\|f_{L,N}-\widetilde{f}_{L,N}\|.
$$
Using \fref{trunc}-\fref{errorFT}-\fref{GDFT} then yields
$$
\|f_{L,N}-\widetilde{f}_{L,N}\| \lesssim  L \|G\|_{L^2(S_L)} \max_{m,n \in I_N} |\hat U_{m,n}-\hat U^N_{m,n}|+ L^2 (\|U\|+L^{-1}\|U\|_{L^2(S_L)}) \left(\frac{C_1}{N^2}+\frac{C^{1-\eps}_1 C^\eps_2}{N^{2(1+\eps)}}\right),
$$
which ends the proof by noticing that $C_1 \lesssim k_0 c_1$ and $C_2 \lesssim k_0 c_2$ for $c_1$ and $c_2$ defined in Appendix \ref{conv}.
\section{Proof of Theorem \ref{th2}}  \label{proofth2}
We start with the error $f_\eta-f_{K,N}$ and split it into three terms using the triangle inequality:
\bee
\| f_\eta-f_{K,N} \| &\leq& \| f_\eta-f_{K}\|+\| f_{K}-\widetilde{f}_{K,N}\|+\| f_{K,N}-\widetilde{f}_{K,N}\|
\eee
where, using Parseval's equality in the second line,
\bee
\| f_\eta-f_{K}\| &\leq &\max_{p,q \in I_N} |(f-f_{K})(x_{p,q})|\\
\| f_{K}-\widetilde{f}_{K,N}\|^2&\leq &N^{-2}\sum_{p,q \in \Zm} |(f_{K}-\widetilde{f}_{K,N})(x_{p,q})|^2 =\frac{1}{L^2} \int_{S_K}|\hat G_\eta(\xi)|^2|\hat U(\xi)-\hat U^N(\xi)|^2 d\xi\\
&\leq & L^{-2} \|\hat G_\eta\|_{L^2(S_K)}^2 \max_{\xi \in S_K}|\hat U(\xi)-\hat U^N(\xi)|^2,
\eee
and
\bee
\| f_{K,N}-\widetilde{f}_{K,N}\| &\leq & \frac{ K^2}{(2 \pi)^2}\max_{p,q \in I_N} |C^N_{p,q}-\widetilde{C}_{p,q}^N|.
\eee
\paragraph{Step 1: the truncation error $f_\eta-f_K$.} First, owing to \fref{Csqrt}, and since the square root function is of H\"older regularity $1/2$, we have, for $|\xi| \geq k_0$,
$$
\big| \Im \sqrt{k_0^2(\eta)-|\xi|^2}-\sqrt{|\xi|^2-k_0^2}\big| \leq \eta/\sqrt{2}.
$$
Hence, 
\bee
|(f_\eta-f_K)(x_{p,q})|&\lesssim &  \max_{|\xi| \geq K} |\hat U(\xi)|\int_K^\infty \frac{e^{- z \Im \sqrt{k_0^2(\eta)-r^2}}}{\sqrt{r^2-k_0^2}} r dr \\
&\lesssim &  \max_{|\xi| \geq K} |\hat U(\xi)|\int_K^\infty \frac{e^{-z \sqrt{r^2-k_0^2}}}{\sqrt{r^2-k_0^2}} r dr \\
&\lesssim &   \max_{|\xi| \geq K} |\hat U(\xi)| \int_{\sqrt{K^2-k_0^2}}^\infty e^{-zu}du\\
&\lesssim  & \max_{|\xi| \geq K} |\hat U(\xi)| \frac{e^{-z\sqrt{K^2-k_0^2}}}{z}
\eee
This is small either when $\hat U$ has negligible contributions when $|\xi| \geq K$, or when $K^2 \geq k_0^2+\frac{a^2}{z^2}$, for some $a$ sufficiently large.

\paragraph{Step 2: the DFT error $C^N_{p,q}-\widetilde{C}_{p,q}^N$.} We can use the aliasing formula since the Green's function $\hat G_\eta$ is smooth thanks to the regularization. The main technicality in the analysis is to obtain (nearly) optimal estimates in the parameter $\eta$. 

Let 
\begin{align} \nonumber
\ell_U=\max &\left(\left(\frac{\max_{i,j=1,2} \max_{\xi \in S_K}| \partial^3_{\xi^2_j \xi_i} \hat U^N(\xi)|}{\max_{\xi \in S_K}  |\hat{U}^N(\xi)|}\right)^{1/3},\left(\frac{\max_{i,j=1,2} \max_{\xi \in S_K}| \partial^2_{\xi_i \xi_j} \hat U^N(\xi)|}{\max_{\xi \in S_K}  |\hat{U}^N(\xi)|}\right)^{1/2}\right.\\
&\hspace{6cm}\left.\frac{\max_{j=1,2} \max_{\xi \in S_K}| \partial_{\xi_j} \hat U^N(\xi)|}{\max_{\xi \in S_K}  |\hat{U}^N(\xi)|} \right). \label{defLU}
\end{align}

We have the following lemma:

\begin{lemma} \label{estH} Let $\eta_0=\eta^2/k_0^2$, $h=L/N$, and for $\ell_U$ and $k_U$ defined in \fref{defLU}-\fref{defKU}, let $\alpha_U=k_U k_0^{-1} \eta_0^{-1/2}$ and $\beta_U=\ell_U +k_Uk_0^{-2} \eta_0^{-3/2}+L k_U k_0^{-1} \eta_0^{-1/2}$, $\gamma_U=\ell_U^2+k_Uk_0^{-3} \eta_0^{-5/2}+L^2 k_U k_0^{-2} \eta_0^{-3/2}$. Define in addition, for any $\eps \in (0,1)$,
  \bee
  D_0&=&\alpha_U + \eta_0^{-1/2-\eps},\qquad  D_1=\beta_U+k_0^{-1} \eta_0^{-3/2-\eps}\\
  D_2&=&\gamma_U+k_0^{-2} \eta_0^{-5/2-\eps}.
  \eee
 Then, for $p,q \in \Zm$,
  \begin{align*}
    K^2 |\widetilde{C}_{p,q}^N| \leq \max_{\xi \in S_K}  |\hat{U}(\xi)|  \Big(k_0^{-2}z^{-1}+D_0^{-1}(|ph|+|qh|)+D_1^{-1}(|ph|^2+|qh|^{2}+|ph| |qh|)\\
    \hspace{3cm} +D_2^{-1}|ph|^{3/2} |qh|^{3/2}\Big)^{-1}.
 \end{align*}
\end{lemma}
\begin{proof} The proof is tedious but not difficult, and consists in estimating integrals of the partial derivatives of $\hat  G_\eta \hat U^N$. We start with estimating $\hat G_\eta$.

  \paragraph{Step 2a: estimates on $\hat G_\eta$.} We denote by $\partial_{|\xi|}$ the radial derivative with respect to $|\xi|$. The following explicit representation of the complex square with positive imaginary part is helpful:
  \be \label{Csqrt}
\sqrt{x+iy}=\frac{1}{\sqrt{2}} \textrm{sign}(y)\sqrt{\sqrt{x^2+y^2}+x}+\frac{i}{\sqrt{2}} \sqrt{\sqrt{x^2+y^2}-x}.
\ee
With $x=k_0^2-|\xi|^2$ and $y=\eta^2$, we find $\Im \sqrt{k_0^2(\eta)-|\xi|^2} \geq \Im \sqrt{k_0^2-|\xi|^2}$. We have also $|\sqrt{k_0^2(\eta)-|\xi|^2}|=((k_0^2-|\xi|^2)^2+\eta^4)^{1/4}\geq |k_0^2-|\xi|^2|^{1/2}$. Using the last two results, we then find
$$
|\hat G_\eta(\xi)| \lesssim \frac{e^{-\Im \sqrt{k_0^2-|\xi|^2}z}}{|\sqrt{k_0^2-|\xi|^2}|},
$$
which is integrable around the singularity at $|\xi|=k_0$. 
The first order radial derivative reads
\bee
\partial_{|\xi|}\hat G_\eta(\xi)&=& 2 i \pi |\xi|\frac{e^{i\sqrt{k_0^2(\eta)-|\xi|^2}z}}{k_0^2(\eta)- |\xi|^2} \left(-z +\big(k_0^2(\eta)- |\xi|^2\big)^{-1/2}\right),\eee
which has non-integrable singularities and we need to use the regularization parameter $\eta$. For this, let $k_0(\eta)=\sqrt{k_0^2+ i \eta^2}$ and set $\eta=k_0 \sqrt{\eta_0}$ for simplicity. It follows directly from \fref{Csqrt} with $x=k_0^2$ and $y=\eta^2$ that $\Re k_0(\eta) \geq k_0 $, and that, for $0 \leq \eta_0 \leq 1$, $\Im k_0(\eta) \geq k_0 \eta_0 / (2 \sqrt{2}) $. These previous inequalities give
$$
|k_0^2(\eta)- |\xi|^2|=|k_0(\eta)-|\xi|||k_0(\eta)+|\xi|| \gtrsim \eta_0k_0 (k_0+|\xi|).
$$
For $0 \leq |\xi| \leq 2 k_0$, $0<\eps<1$, this yields a first bound
$$
|\partial_{|\xi|}\hat G_\eta(\xi)|\lesssim  \frac{|\xi|}{k_0 |k_0-|\xi||^{1-\eps} (\eta_0 k_0)^{\eps}} \left(z +\big(\eta_0 k_0^2\big)^{-1/2}\right),
$$
which has now an integrable singularity at $k_0$. The expression can be simplified using that $z \leq \lambda \leq k_0^{-1}$ and that $\eta_0 <1$. For $|\xi| \geq 2 k_0$, we have $|\xi|^2-k_0^2 \geq 3 |\xi|^2 /4$, so that
$$
|\partial_{|\xi|}\hat G_\eta(\xi)| \lesssim \frac{e^{-\Im \sqrt{k_0^2-|\xi|^2}z}}{|\xi|}\left(z +|\xi|^{-1}\right).
$$
Combining both estimates we find
\be \label{estGH2}
|\partial_{|\xi|}\hat G_\eta(\xi)| \lesssim
\left\{
  \begin{array}{l}
 \ds   \frac{|\xi|}{k_0^2 |k_0-|\xi||^{1-\eps} (\eta_0 k_0)^{\eps} \eta_0^{1/2}} \qquad \textrm{for} \qquad |\xi| \leq 2 k_0\\[5mm]
   \ds  \frac{e^{-\Im \sqrt{k_0^2-|\xi|^2}z}}{|\xi|}\left(z +|\xi|^{-1}\right)  \qquad \textrm{for} \qquad |\xi| \geq 2 k_0.
    \end{array}
  \right.
\ee
Regarding the second-order radial derivative we have
\bee
\partial^2_{|\xi|}\hat G_\eta(\xi)&=& 2 i \pi \frac{e^{i\sqrt{k_0^2(\eta)-|\xi|^2}z}}{k_0^2- |\xi|^2} \left(-z +\big(k_0^2(\eta)- |\xi|^2\big)^{-1/2}\right)\\
&&2  \pi z|\xi|^2\frac{e^{i\sqrt{k_0^2(\eta)-|\xi|^2}z}}{(k_0^2(\eta)- |\xi|^2)^{3/2}} \left(-z +\big(k_0^2(\eta)- |\xi|^2\big)^{-1/2}\right)\\
&&+4 i \pi |\xi|^2\frac{e^{i\sqrt{k_0^2(\eta)-|\xi|^2}z}}{(k_0^2(\eta)- |\xi|^2)^{2}} \left(-z +\big(k_0^2(\eta)- |\xi|^2\big)^{-1/2}\right)\\
&&+2 i \pi |\xi|^2 \frac{e^{i\sqrt{k_0^2(\eta)-|\xi|^2}z}}{(k_0^2(\eta)- |\xi|^2)^{5/2}}
\eee
Proceeding as in the first order case, we find when $0 \leq |\xi| \leq 2 k_0$,
\begin{align*}
&|\partial^2_{|\xi|^2}\hat G_\eta(\xi)|\lesssim \frac{e^{-\Im \sqrt{k_0^2-|\xi|^2}z} }{|k_0(\eta)^2-|\xi|^2|} \times \\
&\left(|z| +\frac{1}{((k_0+|\xi|) \eta_0k_0)^{1/2}}+\frac{z^2|\xi|^2}{((k_0+|\xi|) \eta_0k_0)^{1/2}}+\frac{z |\xi|^2}{((k_0+|\xi|) \eta_0k_0)}+\frac{|\xi|^2}{((k_0+|\xi|) \eta_0k_0)^{3/2}}\right).
\end{align*}
Using that $0 \leq |\xi| \leq 2 k_0$, $z \leq \lambda \leq k_0^{-1}$, and that $\eta_0<1$, this simplifies to
\begin{align*}
&|\partial^2_{|\xi|^2}\hat G_\eta(\xi)|\lesssim \frac{1}{k_0^2 |k_0-|\xi||^{1-\eps} (\eta_0 k_0)^{\eps} \eta_0^{3/2}}.
\end{align*}
When $|\xi| \geq 2 k_0$, we proceed as above, and find for all $|\xi|$,
\be \label{estGH3}
|\partial^2_{|\xi|}\hat G_\eta(\xi)| \lesssim
\left\{
  \begin{array}{l}
 \ds   \frac{1}{k_0^2 |k_0-|\xi||^{1-\eps} (\eta_0 k_0)^{\eps} \eta_0^{3/2}} \qquad \textrm{for} \qquad |\xi| \leq 2 k_0\\[5mm]
   \ds  \frac{e^{-\Im \sqrt{k_0^2-|\xi|^2}z}}{k_0|\xi|}\left(z +|\xi|^{-1}\right)  \qquad \textrm{for} \qquad |\xi| \geq 2 k_0.
    \end{array}
  \right.
\ee



For the third-order derivative, the calculations are very similar and we give the result without proof: we find
\be \label{estGH4}
|\partial^3_{|\xi|}\hat G_\eta(\xi)| \lesssim
\left\{
  \begin{array}{l}
 \ds   \frac{1}{k_0^3 |k_0-|\xi||^{1-\eps} (\eta_0 k_0)^{\eps} \eta_0^{5/2}} \qquad \textrm{for} \qquad |\xi| \leq 2 k_0\\[5mm]
   \ds  \frac{e^{-\Im \sqrt{k_0^2-|\xi|^2}z}}{k_0^2|\xi|}\left(z +|\xi|^{-1}\right)  \qquad \textrm{for} \qquad |\xi| \geq 2 k_0.
    \end{array}
  \right.
\ee

The partials of $\hat G_\eta$ can finally be controlled using the radial derivatives by
\bea \label{par1}
|\partial^2_{\xi_1 \xi_2} \hat G_\eta (\xi)| &\lesssim& |\partial^2_{|\xi|} \hat G_\eta  (\xi)|+\frac{1}{|\xi|}|\partial_{|\xi|} \hat G_\eta  (\xi)|\\ \label{par2}
|\partial^3_{\xi_i^2 \xi_j} \hat G_\eta  (\xi)| &\lesssim&  |\partial^3_{|\xi|} \hat G_\eta  (\xi)|+\frac{1}{|\xi|}|\partial^2_{|\xi|} \hat G_\eta  (\xi)|+\frac{1}{|\xi|^2}|\partial_{|\xi|} \hat G_\eta  (\xi)|.
\eea

We can now turn to the proof of Lemma \ref{estH}. We define $H(\xi)=\hat G_\eta \hat U^N$, and $x=(x_1,x_2)= h(p,q)$  for simplicity. Since $\hat G_\eta$ is even, since $\hat{U}^N$ is periodic, and since $\partial_{\xi_1} \hat G_\eta$ is even w.r.t $\xi_2$, we find after two integrations by parts,
\bea \label{IP1}
H^F(x)&:=&\int_{S_K} H (\xi)e^{i \xi \cdot x}d\xi=-\frac{1}{i x_1} \int_{S_K} \partial_{\xi_1} H (\xi)e^{i \xi \cdot x}d\xi\\  \nonumber 
&=&\frac{1}{x_1^2}\left[\partial_{\xi_1} H (\xi)e^{i \xi \cdot x}\right] -\frac{1}{x_1^2} \int_{S_K} \partial^2_{\xi_1} H (\xi)e^{i \xi \cdot x}d\xi,
\eea
as well as
\bee
H^F(x)&=&-\frac{1}{x_1 x_2} \int_{S_K} \partial^2_{\xi_1 \xi_2} H (\xi)e^{i \xi \cdot x}d\xi\\
&=&\frac{i}{x_1^2 x_2} \left[\partial^2_{\xi_1 \xi_2} H (\xi)e^{i \xi \cdot x}\right]-\frac{i}{x_1^2 x_2} \int_{S_K} \partial^3_{\xi_1^2 \xi_2} H (\xi)e^{i \xi \cdot x}d\xi. 
\eee
Above, the terms in brackets denote the boundary terms coming from the integration by parts. We proceed now to estimating the integrals. We find, using \fref{estGH2}, for all $\eps \in (0,1)$,
\bee
\int_{S_K} |\partial_{|\xi|} \hat G_\eta(\xi)| d \xi &\lesssim& \eta_0^{-1/2-\eps} \int_{r\leq 2} \frac{ dr }{|1-r|^{1-\eps}}\\
&&+\int_{r\geq 2} (z k_0 +r^{-1})e^{-\sqrt{r^2-1} k_0 z}d\xi.
\eee
When $r\geq 2$, we have $r^2-1 \geq 3 r^2 /4$, and using this in the second integral gives
$$
\int_{r \geq k_0 z} (1+r^{-1})e^{-r}dr  \lesssim \log(k_0 z).
$$
Using the assumption that $\log(k_0 z) \lesssim \eta_0^{-1}$, we find
$$
\int_{S_K} |\partial_{|\xi|} \hat G_\eta(\xi)| d \xi \lesssim  \eta_0^{-1/2-\eps}.
$$
Proceeding in the same manner and using \fref{estGH3}-\fref{estGH4}, we find, for all $\eps \in (0,1)$,
\bee
\int_{S_K} |\xi|^{-1}|\partial_{|\xi|} \hat G_\eta(\xi)| d \xi &\lesssim&  k_0^{-1}\eta_0^{-1/2-\eps}\\ \nonumber
\int_{S_K} |\partial^2_{|\xi|} \hat G_\eta(\xi)| d \xi &\lesssim&  k_0^{-1}\eta_0^{-3/2-\eps}\\ \nonumber
\int_{S_K} |\xi|^{-1}|\partial^2_{|\xi|} \hat G_\eta(\xi)| d \xi &\lesssim&  k_0^{-2}\eta_0^{-3/2-\eps}\\ \nonumber
\int_{S_K} |\partial^3_{|\xi|} \hat G_\eta(\xi)| d \xi &\lesssim&  k_0^{-2}\eta_0^{-5/2-\eps}, \nonumber
\eee
which yields, from \fref{par1}-\fref{par2},
\bea \label{int1}
\int_{S_K} |\partial_{\xi_j} \hat G_\eta(\xi)| d \xi &\lesssim&  \eta_0^{-1/2-\eps}\\ \label{int2}
\int_{S_K} |\partial^2_{\xi_j} \hat G_\eta(\xi)|+|\partial^2_{\xi_i \xi_j} \hat G_\eta(\xi)| d \xi &\lesssim&  k_0^{-1}\eta_0^{-3/2-\eps}\\
\int_{S_K} |\partial^3_{\xi_j^2 \xi_i} \hat G_\eta(\xi)| d \xi &\lesssim&  k_0^{-2}\eta_0^{-5/2-\eps}. \label{int3}
\eea
From \fref{IP1} and \fref{int1}, this gives the first estimate
\bee
|x_j| |H^F(x)| &\lesssim& \max_{\xi \in S_K} |\hat G_\eta(\xi)| \int_{S_K}  |\partial_{\xi_j}\hat{U}(\xi)|d\xi + \eta_0^{-1/2-\eps}\max_{\xi \in S_K}  |\hat{U}(\xi)|\\
&\lesssim& \max_{\xi \in S_K}  |\hat{U}(\xi)| ( k_U k_0^{-1} \eta_0^{-1/2}+\eta_0^{-1/2-\eps}),
\eee
since $\max_{\xi \in S_K} |\hat G_\eta(\xi)| \lesssim k_0^{-1} \eta_0^{-1/2}$. We recall that $k_U$ is defined in \fref{defKU}. This gives
\bee
|x_j| |H^F(x)| &\lesssim& \max_{\xi \in S_K}  |\hat{U}(\xi)| ( \alpha_U+\eta_0^{-1/2-\eps})=\max_{\xi \in S_K}  |\hat{U}(\xi)|D_0,
\eee
where $\alpha_U=k_U k_0^{-1} \eta_0^{-1/2}$.

We move on now to the second-order derivative. We find that the boundary term $\left[\partial_{\xi_1} H (\xi)e^{i \xi \cdot x}\right]$ is bounded by
 \be \max_{\xi \in S_K}  |\partial_{\xi_1}\hat{U}(\xi)| +\max_{\xi \in S_K}  |\hat{U}(\xi)| K^{-1} . \label{boundary}\ee
 We have
\bee
\int_{S_K} |\partial^2_{\xi_j} H (\xi)e^{i \xi \cdot x}|d\xi &\leq & \max_{\xi \in S_K} |\hat{U}(\xi)| \int_{S_K} |\partial^2_{\xi_j} \hat G_\eta (\xi)|d\xi\\
&&+\max_{\xi \in S_K}  |\partial_{\xi_j} \hat G_\eta (\xi)| \int_{S_K} |\partial_{\xi_j}\hat{U}(\xi)| d\xi\\
&&+\max_{\xi \in S_K} |\hat G_\eta (\xi)| \int_{S_K}  |\partial^2_{\xi_j}\hat{U}(\xi)| d\xi.
\eee
Using that \fref{boundary}, the fact that $\max_{\xi \in S_K} |\partial_{x_j }\hat G_\eta(\xi)| \lesssim k_0^{-2} \eta_0^{-3/2}$ and \fref{int2}, this gives 
\begin{align*}
  (|x_j|^2+|x_i||x_j|)& |H^F(x)|\\
  &\lesssim \max_{\xi \in S_K}  |\hat{U}(\xi)| (\ell_U +K^{-1}+k_Uk_0^{-2} \eta_0^{-3/2}+L k_U k_0^{-1} \eta_0^{-1/2}+k_0^{-1} \eta_0^{-3/2-\eps}),
\end{align*}
and therefore 
\bee
(|x_j|^2+|x_i||x_j|) |H^F(x)|
&\lesssim &\max_{\xi \in S_K}  |\hat{U}(\xi)| (\beta_U+k_0^{-1} \eta_0^{-3/2-\eps})=\max_{\xi \in S_K}  |\hat{U}(\xi)| D_1,
\eee
where $\beta_U=\ell_U +k_Uk_0^{-2} \eta_0^{-3/2}+L k_U k_0^{-1} \eta_0^{-1/2}$.

With similar calculations involving \fref{int3}, we find with the third-order derivative, 
\bee
|x_i|^2|x_j||H^F(x)|
&\lesssim &\max_{\xi \in S_K}  |\hat{U}(\xi)| (\gamma_U+k_0^{-2} \eta_0^{-5/2-\eps})=\max_{\xi \in S_K}  |\hat{U}(\xi)| D_2,
\eee
where $\gamma_U =\ell_U^2+k_Uk_0^{-3} \eta_0^{-5/2}+L^2 k_U k_0^{-2} \eta_0^{-3/2}$. This ends the proof.
\end{proof}
\medskip

We have now everything needed to quantify the DFT error.
\paragraph{Step 2c: conclusion.}
The aliasing formula gives for $p,q \in I_N$,
 $$
 C^N_{p,q}-\widetilde{C}^N_{p,q}=\sum_{(m,n) \neq (0,0)} \widetilde{C}^N_{p+mN,q+nN}.
 $$
 We break the sum into 3 terms:
 \be \label{3terms}
 C^N_{p,q}-\widetilde{C}^N_{p,q}=\sum_{n \neq 0} \widetilde{C}^N_{p,q+nN}+\sum_{m \neq 0} \widetilde{C}^N_{p+mN,q}+\sum_{m \neq 0}\sum_{n \neq 0} \widetilde{C}^N_{p+mN,q+nN}.
 \ee
 Using Lemma \ref{estH}, the first term is bounded by, for any $\eps' \in (0,1)$, up to the term $\max_{\xi \in S_K}  |\hat{U}(\xi)|/K^2$ that will be reintroduced in the end,
 $$
 \frac{D_0^{1-\eps'} D_1^{\eps'}}{h^{1+\eps'}}\sum_{n \neq 0} \frac{1}{(q+nN)^{1+\eps'}}=\frac{D_0^{1-\eps'} D_1^{\eps'}}{(Nh)^{1+\eps'}}\sum_{n \neq 0} \frac{1}{(\frac{q}{N}+n)^{1+\eps'}} \leq \frac{D_0^{1-\eps'} D_1^{\eps'}}{L^{1+\eps'}},
 $$
 where we used that $q \in I_N$ so that $|q/N|\leq 1/2$. Using the subadditivity of the function $x^\alpha$ for $x>0$ and $\alpha \in (0,1)$, we find
 \bee
 D_0^{1-\eps'} (L^{-1}D_1)^{\eps'} &\leq & \big(\alpha_U^{1-\eps'}+ \eta_0^{-(1/2+\eps)(1-\eps')}\big)\big((L^{-1}\beta_U)^{\eps'}+ (k_0 L)^{-\eps'}\eta_0^{-(3/2+\eps)\eps'}\big)\\
 &\lesssim& (\alpha_U)^{1-\eps'}(k_0 \ell_U^2)^{\eps'}+ k_0^{-\eps'} \eta_0^{-1/2-\eps-\eps'} (1+\eta^{3\eps'/2+\eps \eps'}(k_0 \ell_U)^{2\eps'})).
 \eee
 To simplify the expression above, we remark that $\ell_U \lesssim L$. It is indeed clear that $\ell_U$ is bounded by the support of $U$, which is itself bounded by $L$ by construction. Using the assumption  $(k_0 L)^{-1} \lesssim \eta$, we have $L^{-1}\beta_U \lesssim \alpha_U$. Then,
\bee
 D_0^{1-\eps'} (L^{-1}D_1)^{\eps'} &\leq & \big(\alpha_U^{1-\eps'}+ \eta_0^{-(1/2+\eps)(1-\eps')}\big)\big((\alpha_U)^{\eps'}+ (k_0 L)^{-\eps'}\eta_0^{-(3/2+\eps)\eps'}\big)\\
 &\lesssim& \alpha_U+ \eta_0^{-1/2-\eps}.
\eee
 We obtain the same result for the second term in \fref{3terms}. For the third term, we control it by
\bee\frac{D_1^{1-\eps'} D_2^{\eps'}}{h^{2+\eps'}}\sum_{m \neq 0}\sum_{n \neq 0} \frac{1}{(p+nN)^{1+\eps'}(q+mN)^{1+\eps'}}
\leq \frac{D_1^{1-\eps'} D_2^{\eps'}}{L^{2+\eps'}},
 \eee
 where we find 
 \bee
 (L^{-1}D_1)^{1-\eps'} (L^{-1}D_2)^{\eps'}&\leq & \big((L^{-1}\beta_U)^{1-\eps'}+ (k_0L)^{-(1-\eps')}\eta_0^{-(3/2+\eps)(1-\eps')}\big)\\
 &&\qquad\times \big((L^{-2} \gamma_U)^{\eps'}
 + (k_0L)^{-2\eps'}\eta_0^{-(5/2+\eps)\eps'}\big).
 \eee
As before, using that $(k_0 L)^{-1} \lesssim \eta$ and $\ell_U \lesssim L$, we find $L^{-2}\gamma_U \lesssim \alpha_U$. The above estimate then reduces to
\bee
  (L^{-1}D_1)^{1-\eps'} (L^{-1}D_2)^{\eps'}&\lesssim& \alpha_U+ \eta_0^{-1/2-\eps}.
\eee
Hence, going back to \fref{3terms} and collecting previous estimates, we finally find
\begin{align*}
  &|C^N_{p,q}-\widetilde{C}^N_{p,q}|  \lesssim  K^{-2}L^{-1} \max_{\xi \in S_K}  |\hat{U}(\xi)| \big(\alpha_U+ \eta_0^{-1/2-\eps}\big).
    \end{align*}
 which allows us to estimate the DFT error. It remains now to treat the regularization error to end the proof of Theorem \ref{th2}.

\paragraph{Step 3: the regularization error.}
We write $k_0(\eta)=k_0\sqrt{1+i\eta_0}=k_r+ik_i$, and express $e^{i k_0(\eta) r}-e^{i k_0 r}$ as
$$
e^{i k_0(\eta) r}-e^{i k_0 r}=e^{i k_r r}e^{-k_i r}-e^{i k_0 r}e^{-k_i r}+e^{i k_0 r}e^{-k_i r}-e^{i k_0 r}.
$$
This yields the estimate
$$
\left|e^{i k_0(\eta) r}-e^{i k_0 r}\right| \leq r|k_r-k_0|+rk_i.
$$
Using \fref{Csqrt}, we find
$$
\frac{1}{\sqrt{2}}\sqrt{\sqrt{1+y^2}+1}-1=\frac{1}{2 \sqrt{2}}\int_0^y \frac{x}{\sqrt{1+x^2}\sqrt{\sqrt{1+x^2}+1} } dx \leq \frac{y^2}{4 \sqrt{2}}.
$$
With $y=\eta_0$, this gives $|k_r-k_0| \lesssim k_0 \eta_0^2 $. Regarding the imaginary part, we have when $y \geq 0$,
$$
\sqrt{1+y^2}-1=\int_0^y \frac{xdx}{\sqrt{1+x^2}} \leq \frac{y^2}{2},
$$
so that
$$\frac{1}{\sqrt{2}} \sqrt{\sqrt{1+\eta_0^2}-1} \leq \frac{\eta_0}{2}
$$
and $k_i \lesssim k_0 \eta_0$. As a conclusion, since $\eta_0 \leq 1$, $\left|e^{i k_0(\eta) r}-e^{i k_0 r}\right| \lesssim k_0 r \eta_0$. Then,
\bee
|f(x)-f_\eta(x)| &\leq&  \int_{\Rm^2} |G(x-y)-G_\eta(x-y)||U(y)| dy\\
&\lesssim &  k_0 \eta_0 \int_{\Rm^2}|U(y)| dy.
\eee

This concludes the proof of Theorem \ref{th2}.


\section{Removing evanescent modes} \label{appev} Consider
$$
g(x)=\int_{\Rm^2} \frac{e^{i \sqrt{k_0^2- |\xi|^2} z+i \xi \cdot x}}{\sqrt{k_0^2- |\xi|^2}}\hat U(\xi)d\xi,
$$
which is proportional to the Fourier transform of $f$ (defined in \fref{deff}), neglecting irrelevant constants. The function $\hat U(\xi)$ is smooth, decays to $0$ for $|\xi|\gg k_0$ but is not necessarily integrable. We change variables to arrive at
$$
g(x)=k_0 \int_{\Rm^2} \frac{e^{i \sqrt{1- |\xi|^2} k_0 z+i k_0 \xi \cdot x}}{\sqrt{1- |\xi|^2}}\hat U(k_0\xi) d\xi,
$$
We can use stationary phase techniques to analyze $g$. Suppose first that we are in a far field regime in all directions (i.e. transverse along $x$ and longitudinal along $z$), so that both $k_0 |x|=1/\eps \gg 1$ and $z k_0 \gg 1$. A short calculation shows that there is a  stationary point in $|\xi| <1$ (equal to $\xi=\hat x/(z^2/x^2+1)^{1/2}$, $\hat x=x/|x|$), which gives the leading contribution. So, as expected, the contribution of evanescent modes, corresponding to $|\xi| > 1$, is negligible compared to that of the propagating modes.

In MLB, $z$ is chosen to be small compared to $\min(\lambda,\ell_c)$ for the numerical scheme to be accurate. This means that we are always in a near field regime going from one layer to the next, independently of how large $|x|$ is. In that case, there is no stationary point, and the leading contribution comes from the region around $|\xi|=1$ where the integrand is singular, and there is no particular reason for the evanescent modes to be negligible. We confirm this with the following analysis. Let
$$
g_{\textrm{ev}}(x)=k_0 \int_{|\xi| \geq 1} \frac{e^{-\sqrt{|\xi|^2-1} k_0 z+i k_0 \xi \cdot x}}{i\sqrt{|\xi|^2-1}}\hat U(k_0\xi) d\xi,
$$
be the evanescent modes contribution. To simplify notation and without lack of generality, we suppose that $U$ is real-valued so that $\hat U(-\xi)=\hat U^*(\xi)$. We derive below an asymptotic expression of $g_{\textrm{ev}}$ as $\eps \to 0$. Going to polar coordinates and changing variables, we find
\begin{align*}
g_{\textrm{ev}}(x)&=k_0\int_{0}^{2\pi }\int_{r \geq 1} \frac{e^{-\sqrt{r^2-1} k_0 z+ir \cos \theta /\eps}}{i\sqrt{r^2-1}}\hat U\big(k_0r [\cos\theta \hat x+ \sin \theta \hat x_\perp]\big) r dr d \theta\\
                  &=k_0\int_{0}^{2\pi }\int_{r \geq 0} \frac{e^{-\sqrt{r(r+2)} k_0 z+i(r+1) \cos \theta /\eps}}{i \sqrt{r(r+2)}}\hat U\big(k_0(r+1) [\cos\theta \hat x+ \sin \theta \hat x_\perp]\big) (r+1) dr d \theta.
\end{align*}
Above, $\hat x=x/|x|$ and $\hat x_\perp$ is the vector obtained after a $\pi/2$ rotation of $\hat x$. We then set $r \to \eps r$ and neglect lower order terms in $\eps$, and find, after the change of variables $\cos \theta \to u$ in the second line,
\begin{align*}
  g_{\textrm{ev}}(x) &\simeq\eps^{1/2} k_0 \int_{0}^{2\pi }\int_{r \geq 0} \frac{e^{i\cos \theta /\eps+ir \cos \theta }}{i \sqrt{2r}}\hat U(k_0 [\cos\theta \hat x+ \sin \theta \hat x_\perp]) dr d \theta\\
                     &\simeq \eps^{1/2}k_0\int_{-1}^{1}\int_{r \geq 0} \frac{1}{i \sqrt{2r} \sqrt{1-u^2}}\left[e^{i u/\eps+ir u }\hat U(k_0 [u \hat x+ \sqrt{1-u^2} \hat x_\perp])- cc \right]dr d u.\end{align*}
                   Above, $cc$ means complex conjugate. The integral over $r$ can be performed first, giving
                   $$
                   \int_0^\infty \frac{e^{i r u}}{\sqrt{2r}} dr=\frac{1}{\sqrt{|u|}}F_0(\textrm{sign}(u)):=\frac{1}{\sqrt{|u|}}\int_0^\infty \frac{e^{i r \textrm{sign}(u)}}{\sqrt{2r}} dr.
                   $$
                   After the change of variables $u \to \eps u$ in the second line below and neglecting lower order terms in $\eps$, we finally find                
         \begin{align*}          
                  g_{\textrm{ev}}(x) & \simeq \eps^{1/2}k_0 \int_{-1}^{1}\frac{1}{i \sqrt{|u|}\sqrt{1-u^2}}\left[F_0(\textrm{sign}(u))e^{i u/\eps}\hat U(k_0 [u \hat x+ \sqrt{1-u^2} \hat x_\perp])- cc\right]d u \\
                                     & \simeq -i \eps k_0 [ F_1\hat U(k_0 \hat x_\perp)- cc],       
\end{align*}
where
$$
F_1=\int_\Rm \frac{F_0(\textrm{sign}(u)) e^{i u}}{\sqrt{|u|}} du.
$$
Note that an integration by parts shows that $F_0$ and $F_1$ are well defined. A similar calculation for the propagative modes yields
\begin{align*}  
  g_{\textrm{prop}}(x)&=k_0 \int_{|\xi| \leq 1} \frac{e^{i\sqrt{1-|\xi|^2} k_0 z+i k_0 \xi \cdot x}}{\sqrt{1-|\xi|^2}}\hat U(k_0\xi) d\xi\\
  &\simeq \eps k_0 [ F_2\hat U(k_0 \hat x_\perp)- cc], 
\end{align*}
where $F_2$ is defined like $F_1$ with $F_0$ is replaced by $F_0^*$. This shows that evanescent and propagating modes have nearly identical expressions, which is due to the fact that the leading contribution to $g$ comes from the singularity at $|\xi|=k_0$. Hence, when the field $\hat U$ has a non-zero contribution on $|\xi|=k_0$, removing evanescent modes creates an error as large as the contribution of the propagating modes in the transverse far field regime where $k_0|x| \gg 1$. 

\end{appendix}

\bibliographystyle{siam}
\bibliography{MLB.bib}


\end{document}